\documentclass[12pt,a4paper,oneside]{extarticle}

\usepackage{cmap}	
\usepackage[T2A]{fontenc}
\usepackage[utf8]{inputenc}
\usepackage[english]{babel}

\usepackage{amsfonts}
\usepackage{amsthm}
\usepackage{amstext}
\usepackage{amsmath}
\usepackage{amssymb}
\usepackage{booktabs}
\usepackage{indentfirst} 
\usepackage{enumerate}
\usepackage{caption}[=v1]
\usepackage[colorlinks=true,
            citecolor=blue,
            linkcolor=blue,
            urlcolor=blue]{hyperref}

\usepackage[numbers,square,sort]{natbib} 
\bibpunct{(}{)}{,}{a}{,}{,} 

\usepackage{graphics}
\usepackage[pdftex]{graphicx}
\usepackage{color}
\usepackage{lscape} 
\usepackage{bm} 
\usepackage{tikz}
\usepackage{titlesec} 
\usepackage{float} 
\usepackage{longtable}
\usepackage{booktabs}
\usepackage[center]{subfigure}
\usepackage{captcont}
\usepackage{tabularx}
\usepackage{diagbox} 
\usepackage{rotating}
\usepackage{geometry}
\usepackage[super]{nth}
\usepackage{comment} 
\usepackage{textcomp}
\newcommand{\citeay}[1]{\citeauthor{#1}, \citeyear{#1}}

\newtheorem{ass}{Assumption}[section]
\newtheorem{lem}{Lemma}[section]

\newtheorem{cor}[lem]{Corollary}
\newtheorem{prop}[lem]{Proposition}

\theoremstyle{definition}

\newtheorem{remark}{Remark}[section]
\renewcommand{\baselinestretch}{1.2}

\renewcommand\appendix{\par
  \setcounter{section}{0}
  \setcounter{subsection}{0}
  \setcounter{figure}{0}
  \setcounter{table}{0}
  \renewcommand\thesection{Appendix \Alph{section}}
  \renewcommand\thefigure{\Alph{section}\arabic{figure}}
  \renewcommand\thetable{\Alph{section}\arabic{table}}
  
}

\begin{document}

\renewcommand{\baselinestretch}{1.24}
\title{Robust Cauchy-Based Methods for Predictive Regressions\thanks{We are grateful to the Editor, Anindya Banerjee, and two anonymous referees for their helpful comments and constructive suggestions. We thank Jean-Marie Dufour, Jenny Hau, Nour Meddahi, Aleksey Min, Ulrich K. M{\"u}ller, Neil Shephard, Robert Taylor, and participants at the CFE, NES, iCEBA and iCEBDA conferences for helpful discussions and comments. We also thank Yongok Choi for sharing code. This research was supported, in various and non-overlapping parts, by the Samsung Research Fund through the 2025 SKKU Academic Research Support Program at Sungkyunkwan University (J.~Kim); and by the Basic Research Program at HSE University (A.~Skrobotov).}}

\author{%
Rustam Ibragimov\thanks{New Economic School.}
\and
Jihyun Kim\thanks{Sungkyunkwan University. Corresponding author. E-mail: \href{mailto:kim.jihyun@skku.edu}{kim.jihyun@skku.edu}.}
\and
Anton Skrobotov\thanks{HSE University.}%
}
\date{}
\maketitle


\begin{abstract}
This paper develops robust inference methods for predictive regressions that address key challenges posed by endogenously persistent or heavy-tailed regressors, as well as persistent volatility in errors. Building on the Cauchy estimation framework, we propose two novel tests: one based on $t$-statistic group inference and the other employing a hybrid approach that combines Cauchy and OLS estimation. \textcolor{black}{The tests target robust inference on individual predictive coefficients and extend to a conservative test of joint predictability.} These methods effectively mitigate size distortions that commonly arise in standard inference procedures under endogeneity, near nonstationarity, heavy tails, and persistent volatility. The proposed tests are simple to implement and applicable to both continuous- and discrete-time models. Extensive simulation experiments demonstrate favorable finite-sample performance across a range of realistic settings. An empirical application examines the predictability of excess stock returns using the dividend–price and earnings–price ratios as predictors. The results suggest that the dividend–price ratio possesses predictive power, whereas the earnings–price ratio does not significantly forecast returns.
\end{abstract}



\smallskip
 \noindent \emph{Word count}: 12,471



\newpage
\renewcommand{\baselinestretch}{1.4}

\section{Introduction}

Predictive regressions play a central role in empirical finance, providing a framework for assessing whether financial or macroeconomic variables can forecast future returns. Prominent applications include the forecasting of equity and aggregate returns (see, among others, \citeay{CampbellYogo2006}; \citeay{GoyalWelch}; \citeay{campbell}; \citeay{Hirshleifer}; \citeay{KJ}; \citeay{Rapach}; \citeay{MR}; \citeay{Goyal2024}, and references therein) and tests of market efficiency (e.g., \citeauthor{Fama1}, \citeyear{Fama1, Fama, Fama2}; the review in \citeay{Martin}, and references therein). A large body of work has examined the econometric properties of predictive regressions (see \citeay{Phillips2015} for a review), highlighting several statistical challenges that complicate inference on return predictability.

A key difficulty arises from the persistence and endogeneity of commonly used predictors. Variables such as the dividend--price and earnings--price ratios typically exhibit near–unit–root behavior, and their innovations are correlated with future returns over long horizons. This combination induces substantial biases in conventional hypothesis tests (see, e.g., \citeauthor{stambaugh1999predictive}, \citeyear{stambaugh1999predictive}; \citeauthor{KimPark1}, \citeyear{KimPark1}). In addition, stock return volatility is stochastic and highly persistent (\citeauthor{jacquier2004bayesian}, \citeyear{jacquier2004bayesian}; \citeauthor{hansen2014estimating}, \citeyear{hansen2014estimating}), and \citet{cavaliere2004testing} shows that such volatility can lead to severe size distortions in tests that rely on stationarity. Predictive regression data also often exhibit heavy tails, jumps, structural breaks, and regime switching, further undermining standard inference methods.

A large literature has proposed procedures to address persistent endogeneity in predictive regressions. These include Bonferroni-based confidence-set methods, bias-correction and projection approaches, empirical-likelihood and variable-addition methods, IVX procedures, and Cauchy-type sign-instrument methods (see, among others, \citeauthor{CampbellYogo2006}, \citeyear{CampbellYogo2006}; \citeauthor{ChenDeo2009}, \citeyear{ChenDeo2009}; \citeauthor{PhillipsMagdalinos2009}, \citeyear{PhillipsMagdalinos2009}; \citeauthor{KMS2015}, \citeyear{KMS2015}; \citeauthor{CJP2016}, \citeyear{CJP2016}; \citeauthor{Liu2019}, \citeyear{Liu2019}; and the surveys in \citeauthor{Phillips2015}, \citeyear{Phillips2015}). However, many of these approaches rely on assumptions such as homoskedasticity or weak dependence in volatility and may not remain valid under empirically relevant features such as persistent or nonstationary volatility, structural breaks, or regime switching. For example, \citet{ibragimov2023new} show that standard tests can suffer from severe size distortions in the presence of persistent volatility.

To address these limitations, \citet{CJP2016} propose an inference method (the \emph{Cauchy RT}) based on the Cauchy estimator and a time-change transformation in a continuous-time framework.\footnote{See also \citet{BuKimWang} for an alternative method robust to endogenously persistent or heavy-tailed regressors and persistent volatility in continuous time.} The appeal of the Cauchy estimator is that the sign instrument removes the dependence of the normalized limit on the persistence and tail behavior of the regressor; related developments and applications include \citet{BREITUNG2015}, \citet{kim-meddahi-2019}, and \citet{shephard2020}. \citet{ibragimov2023new} introduce another approach (the \emph{Cauchy VC}), also based on the Cauchy estimator but with a nonparametric volatility correction, which applies to both continuous- and discrete-time models.

This paper proposes two practical tests that serve as robust alternatives to these methods. \textcolor{black}{Our principal contribution is robust inference on an individual predictive coefficient, together with a conservative test of joint predictability for multiple predictors.} The proposed procedures are robust to heterogeneous and persistent volatility, as well as to endogenous, persistent, and heavy-tailed regressors. Both methods employ Cauchy estimation, as in \citet{CJP2016} and \citet{ibragimov2023new}, to address endogeneity, persistence, and heavy tails. The two approaches differ in their treatment of heterogeneous volatility: the first extends the $t$-statistic–based group inference of \citet{IbragimovMuller2010} to asymptotically normal Cauchy estimators, while the second is a hybrid procedure that combines Cauchy and OLS estimation by using the Cauchy estimator for the coefficient and OLS residuals for variance estimation.

The proposed methods are simple to implement and avoid the technical complexities associated with existing approaches, such as the time-change transformation in \citet{CJP2016} or the nonparametric volatility correction in \citet{ibragimov2023new}. Although they rely on an asymptotic exogeneity condition on volatility, simulation results show that they perform well in finite samples, even under mild violations of this condition. Moreover, the proposed methods apply to both continuous- and discrete-time models. Overall, our procedures complement existing Cauchy-based methods and provide a practical alternative for robust inference under a broad range of empirically relevant environments.

\textcolor{black}{The multivariate extension is more limited in scope. The robustness properties described above concern inference on a single predictive coefficient. With multiple predictors, a Bonferroni combination of the univariate tests yields a joint predictability test that retains this robustness, because it imposes no assumption on the dependence across predictors. Its costs are conservativeness and the lack of conditional inference on an individual coefficient. In our framework, genuinely multivariate inference, which exploits the joint covariance of the predictors, relies on a Wald-type extension. This extension requires substantially stronger assumptions -- consistent estimation of the joint covariance and, for persistent regressors, asymptotic exogeneity -- and hence does not retain the full robustness of the univariate procedures. In its present form, it also tests only the joint null; conditional inference on individual coefficients or subsets is left for future research. Moreover, accommodating a nonzero intercept, which is standard in empirical work, requires first differencing; this preserves validity but reduces power, reinforcing the practical limitations of the multivariate extension. For flexible joint and conditional inference with persistent, endogenous regressors and intercepts, the IVX approach of \citet{KMS2015} is the natural tool, although it does not target robustness to heavy tails or persistent volatility.}

\textcolor{black}{We also evaluate our methods in discrete-time predictive regressions and compare them with the IVX approach of \citet{KMS2015}.} The IVX method is specifically designed to address persistence and endogeneity and performs well under homoskedastic or weakly dependent volatility. Its construction, however, depends on the persistence-filtering device used to generate the IVX instruments, and the associated reduction in convergence rates may entail power costs, a point emphasized in the IVX and predictive-quantile literature (see, e.g., \citeauthor{Liu2019}, \citeyear{Liu2019}; \citeauthor{CaiChenLiao2023}, \citeyear{CaiChenLiao2023}). Our results highlight that the two approaches are complementary. While IVX delivers strong performance under standard conditions, it may suffer from size distortions in the presence of nonstationary or persistent volatility. By contrast, our methods are designed to remain robust to such features, including heavy tails and time-varying volatility.

The remainder of the paper is organized as follows. Section~2 introduces the predictive regression model and the Cauchy estimator. Section~3 develops the proposed inference procedures and establishes their theoretical properties. Section~4 extends the analysis to models with multiple predictors and intercepts. Sections~5 and~6 present simulation results and an empirical application. Section~7 concludes. All proofs are collected in Appendix~D of the accompanying Online Appendix, which also reports additional theoretical analysis and supplementary simulation and empirical results.

\section{Predictive Regressions and the Cauchy Estimator}

\subsection{Model and Issues}

Throughout the paper, we consider $(\mathcal{F}_t)$-adapted processes defined on a filtered probability space $(\Omega,\mathcal{F}, (\mathcal{F}_t), P)$ equipped with an increasing filtration $(\mathcal{F}_t)$ of sub-$\sigma$-fields of $\mathcal{F}$. Our objective is to test the (un)predictability of the process $(y_t)$ (e.g., the time series of excess stock returns) based on a covariate process $(x_t)$ (e.g., the time series of price–dividend ratios). As usual, we consider the linear predictive regression model
\begin{align}
y_t &= \alpha + \beta x_{t-1} + u_t, \quad t=1, \dots, T. \label{PredRegr1}
\end{align}

Following the standard specification for a volatility model, we assume that
\begin{align*}
u_t = v_t \varepsilon_t,
\end{align*}
where $(v_t)$ is a volatility process and $(\varepsilon_t)$ is a martingale difference sequence (MDS) with respect to $(\mathcal{F}_t)$. We impose the following regularity conditions on $(v_t, \varepsilon_t)$.

\begin{ass}\label{assumption-mds}
(a) $E(\varepsilon_t^2|\mathcal{F}_{t-1})=1$; 
(b) $(v_t)$ is $(\mathcal{F}_{t-1})$-adapted and nonnegative; and 
(c) $T^{-1}\sum_{t=1}^T E\!\left[\varepsilon_t^2 1\{|\varepsilon_t|\geq \delta \sqrt{T}\}\middle|\mathcal{F}_{t-1}\right]\!\to_p 0$ for any $\delta>0$.
\end{ass}

Conditions (a) and (b) are standard and ensure that the conditional variance of $u_t$ is identified: $E(u_t^2|\mathcal{F}_{t-1})=v_t^2$. Condition (c) is a conditional Lindeberg condition, which holds, for example, if $\sup_{1\le t\le T} E(|\varepsilon_t|^{2+\delta}|\mathcal{F}_{t-1})$ is bounded for some $\delta>0$ with probability one. See \citet{ibragimov2023new} and references therein for further discussion and examples of processes satisfying Assumption~\ref{assumption-mds}.

The hypothesis of unpredictability of $(y_t)$ corresponds to $H_0:\beta=0$ in regression~\eqref{PredRegr1}. It is well-known that standard OLS $t$-statistic inference is not robust to many empirically relevant features of financial data. For instance, the OLS estimator of $\beta$ is not asymptotically Gaussian under $H_0$ if $(x_t)$ is endogenous and (nearly) nonstationary (see \citeay{elliott1994inference}; \citeay{PhillipsNear}; \citeay{GP}; \citeay{PM}; \citeay{KMS2015}), or even if $(x_t)$ is stationary but has infinite variance (e.g., \citeay{GrangerOrr}; \citeay{EKM}; \citeay{ibragimov2015heavy}). This non-Gaussianity persists even when the errors are homoskedastic with $v_t^2=\sigma^2$ for all $t$.\footnote{As usual, the endogeneity of $(x_{t-1})$ refers to the existence of nonzero long-run covariance between the innovations of $(y_t)$ and $(x_{t-1})$.} Furthermore, stock return data exhibit time-varying and stochastically persistent volatility, which causes the distribution of the OLS $t$-statistic to deviate from standard normality, leading to size distortions in conventional tests (see \citeauthor{CJP2016}, \citeyear{CJP2016}; \citeauthor{ibragimov2023new}, \citeyear{ibragimov2023new}).  

\subsection{The Cauchy Estimator}

Both inference methods proposed in this paper build upon the following \emph{Cauchy estimator} of $\beta$ (assuming no intercept, i.e., $\alpha=0$):
\begin{align*}
\check{\beta} 
= \left(\sum_{t=1}^T |x_{t-1}| \right)^{-1} \sum_{t=1}^T \text{sign}(x_{t-1})\, y_t,
\end{align*}
where $\text{sign}(\cdot)$ denotes the sign function, $\text{sign}(x)=1$ for $x\ge0$ and $\text{sign}(x)=-1$ for $x<0$. The estimator $\check{\beta}$ can be interpreted as an instrumental variable (IV) estimator using $\text{sign}(x_{t-1})$ as an instrument for $x_{t-1}$ (see, e.g., \citeay{SoShin}; \citeay{BREITUNG2015}; \citeay{kim-meddahi-2019}; \citeay{shephard2020}). 

The no-intercept formulation should be interpreted carefully. In \citet{CJP2016} and in several subsequent Cauchy-based or sign-instrument approaches (e.g., \citet{SoShin}; \citet{BREITUNG2015}; \citet{kim-meddahi-2019}; \citet{shephard2020}), the predictive regression is often written without a separate scalar intercept, or the intercept is absorbed into the regressor vector by including a constant component. That convention is convenient for presenting the Cauchy estimator as an IV estimator. In this paper we keep the distinction explicit: Sections~2 and~3 analyze the scalar no-intercept benchmark, while Section~4.2 treats the empirically standard case with a nonzero intercept. This separation is useful because demeaning is not innocuous for the martingale arguments underlying the Cauchy statistic, whereas the first-differencing approach in Section~4.2 restores the relevant martingale-difference structure.

Under Assumption~\ref{assumption-mds}, $\text{sign}(x_{t-1})\varepsilon_t$ (denoted by $\xi_t$) is a unit-variance MDS with respect to $(\mathcal{F}_t)$. Define the continuous-time partial sum process $(W^T(r), 0\le r\le1)$ by
\[
W^T(r) = T^{-1/2} \sum_{t=1}^{[Tr]} \xi_t,
\]
which takes values in $\mathbf{D}_{\mathbb{R}}[0,1]$, the space of c\`adl\`ag functions on $[0,1]$ with values in \textcolor{black}{$\mathbb{R}$}. By the functional central limit theorem for martingales (Theorem~18.2 of \citeauthor{billingsley1986convergence}, \citeyear{billingsley1986convergence}), we have $W^T \Rightarrow W$ in $\mathbf{D}_{\mathbb{R}}[0,1]$, where $W$ is a standard Brownian motion.  

For the volatility process $(v_t)$, define $\sigma^T(r)=v_{[Tr]}$ on $\mathbf{D}_{\mathbb{R}^+}[0,1]$. Then the Cauchy estimator can be expressed in terms of $\sigma^T$ and $W^T$ as
\begin{align*}
\left(\sum_{t=1}^T |x_{t-1}|/\sqrt{T}\right)\! \bigl(\check{\beta}-\beta\bigr)
&= T^{-1/2}\sum_{t=1}^T \text{sign}(x_{t-1})v_t\varepsilon_t 
= \int_0^1 \sigma^T(r)\, dW^T(r).
\end{align*}

Following \citet{ibragimov2023new}, we assume that the volatility process $\sigma^T$ is persistent in the sense that it admits a limiting process $\sigma$ defined on $[0,1]$ such that $(W^T,\sigma^T)\Rightarrow(W,\sigma)$ jointly.

\begin{ass}\label{assumption-1-2}
There exists a nonnegative process $\sigma$ on $\mathbf{D}_{\mathbb{R}^+}[0,1]$ such that 
\[
0<\int_0^1 \sigma^2(r)\,dr<\infty, 
\quad\text{and}\quad
(W^T,\sigma^T)\Rightarrow(W,\sigma)
\quad\text{in}\quad \mathbf{D}_{\mathbb{R}\times\mathbb{R}^+}[0,1],
\]
where $W$ is a standard Brownian motion adapted to the same filtration as $\sigma$.
\end{ass}

Assumption~\ref{assumption-1-2} encompasses a wide class of models, including those with nonstationary volatility, regime switching, or structural breaks.\footnote{Assumptions~\ref{assumption-mds} and~\ref{assumption-1-2} exclude some globally homoskedastic processes, such as stationary GARCH models. However, the hybrid testing procedure proposed later remains valid under $T^{-1}\sum_{t=1}^T v_t^2\to_p\omega^2>0$, which includes conditionally heteroskedastic but globally homoskedastic processes, such as stationary GARCH models (see also Section~4 of \citeauthor{ibragimov2023new}, \citeyear{ibragimov2023new}).} It also covers cases with deterministic volatility $v_t=\sigma(t/T)$, as in \citet{cavaliere-taylor-2007,cavaliere-taylor-2008}, \citet{xu-phillips-2008}, and \citet{harvey-leybourne-zu}, among others.\footnote{Assumption~\ref{assumption-1-2} is a simplified version of the condition $v_{[Tr]}/a_T\Rightarrow\sigma_r$ in Assumption~2 of \citet{cavaliere-taylor-2009}. We focus on stochastically bounded volatilities with $a_T=1$, excluding explosive volatility settings ($a_T\to\infty$) for simplicity.} It further includes nonstationary volatility processes such as those in \citet{hansen1995regression} and \citet{chung2007nonstationary}, where $v_t^2$ is a smooth positive transformation of a (near) unit root process. Overall, Assumptions~\ref{assumption-mds} and~\ref{assumption-1-2} are general enough to allow for stochastic and discontinuous volatility—features commonly observed in financial returns.

Under Assumptions~\ref{assumption-mds} and~\ref{assumption-1-2}, the properly normalized Cauchy estimator satisfies
\[
\left(\sum_{t=1}^T |x_{t-1}|/\sqrt{T}\right)\!(\check{\beta}-\beta)
\Rightarrow \int_0^1 \sigma(r)\,dW(r),
\]
by standard results on the convergence of stochastic integrals (see \citeauthor{hansen-1992}, \citeyear{hansen-1992}; \citeauthor{kurtz1991weak}, \citeyear{kurtz1991weak}; \citeauthor{IP}, \citeyear{IP}; \citeauthor{ibragimov2023new}, \citeyear{ibragimov2023new}). The limit $\int_0^1 \sigma(r)\,dW(r)$ is in general a non-Gaussian martingale, becoming Gaussian only if $W$ and $\sigma$ are independent. In that case, $\int_0^1 \sigma(r)\,dW(r)$ is a scale mixture of normals with variance $\int_0^1 \sigma^2(r)\,dr$, denoted
\[
\int_0^1 \sigma(r)\,dW(r)\;=_d\; \mathbb{MN}\!\left(0, \int_0^1 \sigma^2(r)\,dr\right).
\]
We formalize the independence assumption as follows.

\begin{ass}\label{assumption-2-1}
The processes $\sigma$ and $W$ in Assumption~\ref{assumption-1-2} are independent.
\end{ass}

Assumption~\ref{assumption-2-1} requires the volatility process $\sigma^T$ to be asymptotically independent of the martingale $W^T$. It does not, however, require exact independence in finite samples: it permits finite-sample dependence between $\sigma^T$ and $W^T$ provided such dependence is asymptotically negligible. This is the empirically relevant case. For the limiting theory to approximate finite-sample behavior accurately, any finite-sample dependence between volatility and the martingale innovations must be weak, that is, a secondary feature of the data rather than a dominant one; Assumption~\ref{assumption-2-1} formalizes precisely this requirement. For example, consider
\[
\sigma^T\!\left(t/T\right) 
= T^{-\delta} f(x_{t-1}, \varepsilon_t) + \sigma_0^T\!\left(t/T\right), \quad \delta>0,
\]
where $f:\mathbb{R}^2\to\mathbb{R}^+$ is bounded and $\sigma_0^T$ is independent of $W^T$ with $(W^T,\sigma_0^T)\Rightarrow(W,\sigma)$, where $W$ and $\sigma$ are independent. Here the dependence between $\sigma^T$ and $W^T$ enters only through the term $T^{-\delta} f(x_{t-1},\varepsilon_t)$, which is of order $T^{-\delta}$ and hence asymptotically negligible for any $\delta>0$; thus $\sigma^T$ satisfies Assumption~\ref{assumption-2-1} while remaining dependent on $W^T$ in finite samples. The exponent $\delta$ measures the strength of this finite-sample dependence -- a larger $\delta$ corresponds to faster decay and weaker dependence -- so the example covers precisely the empirically relevant case in which volatility and the martingale innovations are at most weakly dependent.

In the following sections, we develop inference methods based on the Cauchy estimator. Section~3 focuses on predictive regressions with a single predictor and no intercept, while Section~4 extends the analysis to models with multiple predictors and an intercept.

\section{Robust Inference for Predictive Regressions}

\subsection{Robust \texorpdfstring{$t$}{t}-Statistic Approach}

The first approach relies on $t$-statistic-based inference using group estimates of $\beta$, as proposed by \citet{IbragimovMuller2010} (see also \citeauthor{IM1}, \citeyear{IM1}; Section~3.3 of \citeauthor{ibragimov2015heavy}, \citeyear{ibragimov2015heavy}). The method is based on normalized Cauchy estimators—specifically, the numerator of the Cauchy estimator divided by $\sqrt{T}$ in the full-sample case:
\begin{equation}\label{numCauchy}
\check{\gamma} = \frac{\sum_{t=1}^T |x_{t-1}|}{\sqrt{T}}\check{\beta}
= \frac{1}{\sqrt{T}}\sum_{t=1}^T \text{sign}(x_{t-1})y_t.
\end{equation}

Following the $t$-statistic approach, we partition the sample into a fixed number $q\ge2$ of approximately equal groups of consecutive observations. The observation $(y_t,x_{t-1})$ at time $t$ belongs to the $j$th group $\mathcal{G}_j$ if
\[
t\in\mathcal{G}_j = \{s:(j-1)[T/q]<s\le j[T/q]\}, \quad j=1,\dots,q.
\]
We compute the normalized Cauchy statistic in \eqref{numCauchy} within each group:
\begin{equation}\label{gammagroup}
\check{\gamma}_j = \sqrt{\frac{q}{T}}\sum_{t=(j-1)[T/q]+1}^{j[T/q]}
\text{sign}(x_{t-1})y_t.
\end{equation}

The $t$-statistic based on the $q$ group statistics $\{\check{\gamma}_j\}_{j=1}^q$ is given by
\begin{equation}\label{IM2010}
t_q(\check{\gamma}) = \sqrt{q}\,\frac{\bar{\gamma}}{s_{\gamma}},
\end{equation}
where
\[
\bar{\gamma} = q^{-1}\sum_{j=1}^q \check{\gamma}_j,
\qquad
s_{\gamma}^2 = (q-1)^{-1}\sum_{j=1}^q (\check{\gamma}_j-\bar{\gamma})^2.
\]
Under the null hypothesis $H_0:\beta=0$, the test rejects $H_0$ in favor of $H_A:\beta\neq0$ if
$|t_q(\check{\gamma})|>cv_q(\alpha)$, where $cv_q(\alpha)$ denotes the two-sided $t$-critical value at level $\alpha$, i.e.\ $P(|T_{q-1}|>cv_q(\alpha))=\alpha$ for $T_{q-1}\sim t_{q-1}$ (one-sided tests are analogous).

To study the asymptotic behavior of $\{\check{\gamma}_j\}_{j=1}^q$, we decompose
\[
\check{\gamma}_j = \zeta_j + \psi_j,
\]
where
\[
\zeta_j = \beta\sqrt{\frac{q}{T}}\sum_{t\in\mathcal{G}_j}|x_{t-1}|,
\qquad
\psi_j = \sqrt{\frac{q}{T}}\sum_{t\in\mathcal{G}_j}\text{sign}(x_{t-1})u_t.
\]
Under Assumption~\ref{assumption-mds}, $\{\psi_j\}_{j=1}^q$ forms a sequence of martingale differences uncorrelated across groups, yielding the following asymptotic characterization.

\begin{lem}\label{lemma-2-1}
Let Assumptions~\ref{assumption-mds}, \ref{assumption-1-2}, and \ref{assumption-2-1} hold. For any fixed $q\ge2$ and $\beta\in\mathbb{R}$,
\[
(\check{\gamma}_1-\zeta_1,\dots,\check{\gamma}_q-\zeta_q)'
\to_d
\mathbb{MN}\bigl(0,\mathrm{diag}(q\omega_1^2,\dots,q\omega_q^2)\bigr),
\]
where $\omega_j^2 = \int_{(j-1)/q}^{j/q}\sigma^2(r)\,dr$ for $j=1,\dots,q$.
\end{lem}

The statistics $\{\check{\gamma}_j\}_{j=1}^q$ do not satisfy the standard condition in \citet{IbragimovMuller2010}, which requires estimators $\{\tilde{\beta}_j\}_{j=1}^q$ such that
\[
\{m_T(\tilde{\beta}_j-\beta)\}_{j=1}^q \to_d \{V_j Z_j\}_{j=1}^q,
\]
for some $m_T\to\infty$, $Z_j\stackrel{iid}{\sim}\mathbb{N}(0,1)$, and $\{V_j\}$ independent of $\{Z_j\}$. By contrast, Lemma~\ref{lemma-2-1} shows that $\{\check{\gamma}_j\}_{j=1}^q$ lack such a diverging normalization. Consequently, as shown in Proposition~\ref{proposition-2-1}, the $t$-statistic approach yields correct asymptotic size but is consistent only for a restricted class of covariates, excluding (near) unit-root processes. This inconsistency arises precisely because the asymptotics of $\{\check{\gamma}_j\}_{j=1}^q$ do not involve a diverging sequence (see proofs of Proposition~\ref{proposition-2-1} and Corollary~\ref{corollary-2-1}).

Nevertheless, with additional regularity conditions, if $\{c_T^{-1}\sum_{t\in\mathcal{G}_j}|x_{t-1}|\}_{j=1}^q\to_d\{D_j\}_{j=1}^q$ for positive random variables $\{D_j\}$ and a sequence $c_T/\sqrt{T}\to\infty$, then the Cauchy estimator $\check{\beta}_j$ computed within each group satisfies
\[
\{m_T(\check{\beta}_j-\beta)\}_{j=1}^q \to_d \{P_j\}_{j=1}^q,
\]
for $m_T=c_T\sqrt{q/T}$. In general, however, $\{P_j\}_{j=1}^q$ are non-Gaussian, especially when $(x_t)$ is (near) unit root and endogenous. Applying the $t$-statistic approach to $\{\check{\beta}_j\}_{j=1}^q$ thus yields consistency for broader classes of covariates but may incur size distortions due to non-Gaussianity.

\begin{prop}\label{proposition-2-1}
Let Assumptions~\ref{assumption-mds}, \ref{assumption-1-2}, and \ref{assumption-2-1} hold, with fixed $q\ge2$ and $\alpha\le\textcolor{black}{0.083}$.

(a) Under $H_0:\beta=0$,
\[
\lim_{T\to\infty}\mathbb{P}(|t_q(\check{\gamma})|>cv_q(\alpha)\mid H_0)\le\alpha.
\]

(b) Under $H_A:\beta\neq0$, suppose $(x_t)$ is stationary with $E|x_t|<\infty$ and satisfies
\[
\sup_{1\le s\le T-T/q}\Biggl|E|x_t|-\frac{1}{T/q}\sum_{t=s}^{s+T/q}|x_t|\Biggr|\to_p0.
\]
Then
\[
\lim_{T\to\infty}\mathbb{P}(|t_q(\check{\gamma})|>cv_q(\alpha)\mid H_A)=1.
\]
\end{prop}

\begin{figure}[ht]
\centering
\includegraphics[scale=0.18]{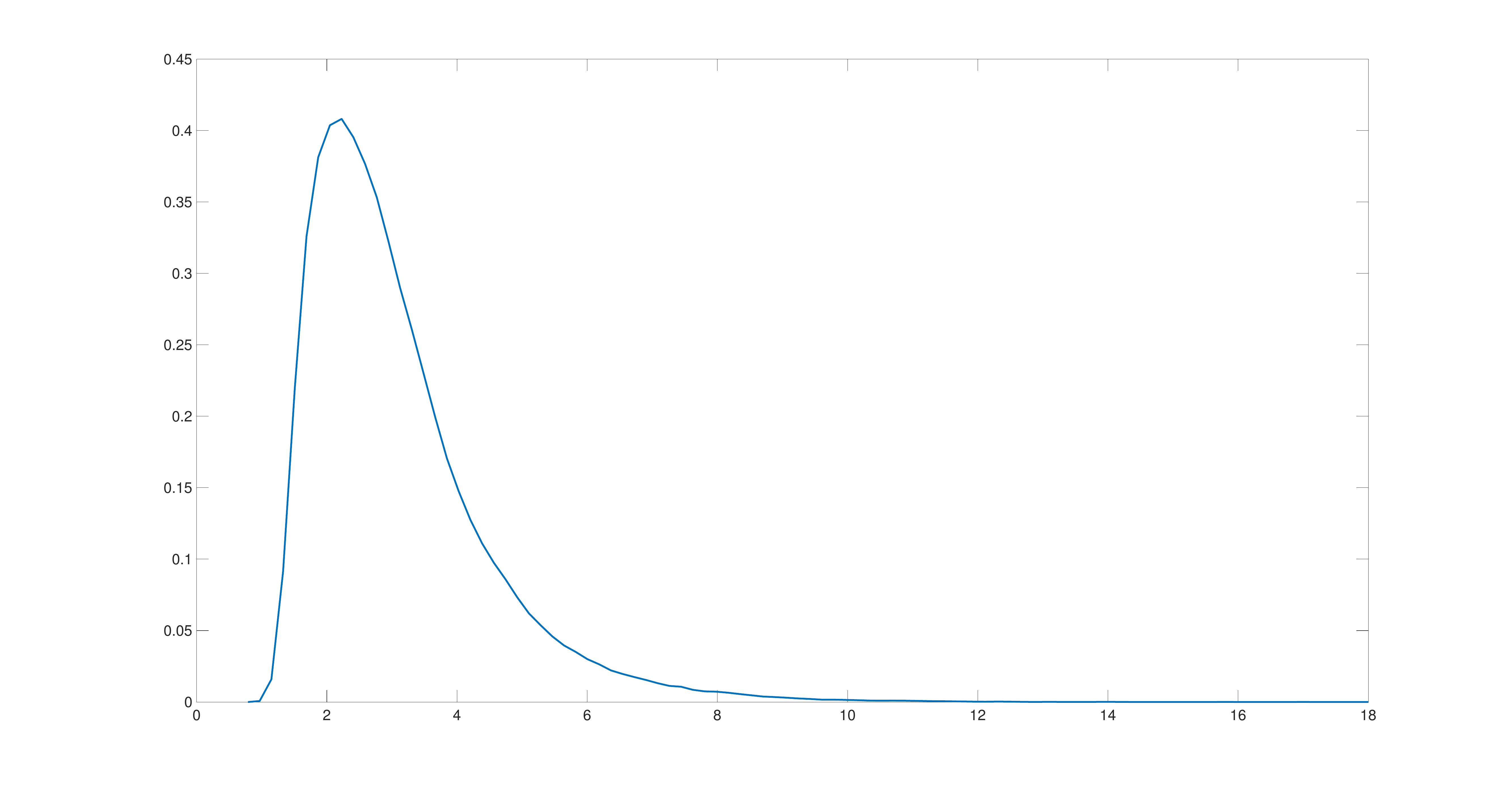}
\caption{Simulated density of $\mathcal{D}_2$ in \eqref{simul}.}
\label{simulfig}
\end{figure}

Proposition~\ref{proposition-2-1} shows that the $t$-statistic approach is conservative under $H_0$ and consistent under $H_A$ when $(x_t)$ is stationary with a finite first moment. It is thus valid and robust to persistent heteroskedasticity and endogenously heavy-tailed covariates. However, it becomes inconsistent for highly persistent covariates, such as (near) unit-root processes. To illustrate, consider the generalized local-to-unity framework of \citet{dou2021generalized}, where $X^T(r)=x_{[Tr]}$ for $r\in[0,1]$ and
\begin{equation}\label{fclt}
T^{-1/2}\bigl(X^T(\cdot)-X^T(0)\bigr)\to_d X(\cdot)-X(0),
\end{equation}
with $X$ a stationary continuous-time Gaussian ARMA process.\footnote{See \citet{dou2021generalized} for a detailed discussion.}

\begin{cor}\label{corollary-2-1}
Let Assumptions~\ref{assumption-mds}–\ref{assumption-2-1} hold and suppose $(x_t)$ satisfies \eqref{fclt}. Under $\beta\neq0$, $t_q(\check{\gamma})\to_d\text{sign}(\beta)\mathcal{D}_q$ for $q\ge2$, where
\[
\mathcal{D}_q =
\int_0^1 |X(r)|dr
\left(
\frac{q(q-1)}{\sum_{j=1}^q\bigl(\int_0^1|X(r)|dr - q\int_{(j-1)/q}^{j/q}|X(r)|dr\bigr)^2}
\right)^{1/2},
\]
and $\mathcal{D}_q>(q-1)^{-1/2}$ with probability one.
\end{cor}

When $(x_t)$ is highly persistent, $t_q(\check{\gamma})$ converges to $\mathcal{D}_q$ rather than diverging, with lower bound $(q-1)^{-1/2}$. Simulations in Section~5 confirm that rejection probabilities remain high even when $t_q(\check{\gamma})$ is asymptotically bounded. For $q=2$,
\begin{align}\label{simul}
\mathcal{D}_2 =
\frac{\int_0^1|X(r)|dr}
{\bigl|\int_0^{1/2}|X(r)|dr - \int_{1/2}^1|X(r)|dr\bigr|} > 1.
\end{align}
The ratio form in \eqref{simul} implies large realizations of $\mathcal{D}_2$ in finite samples, producing high rejection rates even under inconsistency. Figure~\ref{simulfig} plots the simulated density of $\mathcal{D}_2$ when $X$ is Brownian motion.\footnote{Based on 100{,}000 simulated draws.} The minimum simulated value is 1.15, and $\mathbb{P}(|\mathcal{D}_2|>cv_2(0.05))=0.15$ with $cv_2(0.05)=4.303$.

\subsection{A Hybrid Test}

We now propose a simple hybrid test that remains consistent for a broad class of covariates. Under Assumptions~\ref{assumption-mds}–\ref{assumption-2-1},
\[
\frac{\sum_{t=1}^T |x_{t-1}|}{\sqrt{T}}(\check{\beta}-\beta)
= \frac{1}{\sqrt{T}}\sum_{t=1}^T \text{sign}(x_{t-1})u_t
\to_d \int_0^1\sigma(r)\,dW(r)=\omega Z,
\]
where $Z\sim\mathbb{N}(0,1)$ and $\omega^2=\int_0^1\sigma^2(r)\,dr$.

A key feature of the Cauchy estimator $\check{\beta}$ is that its properly normalized limit distribution is invariant to the data-generating process of $(x_t)$. By contrast, the OLS estimator’s variance depends on both $(u_t)$ and $(x_t)$, complicating variance estimation even under homoskedasticity. For $\check{\beta}$, the asymptotic variance depends solely on $u_t$, requiring only heteroskedasticity-robust adjustments.\footnote{See \citet{shephard2020}, Section~4.3, for related discussion.}

We define the hybrid test statistic as
\[
\tau(\check{\beta}) = \frac{\check{\gamma}}{\hat{\omega}},
\]
where $\check{\gamma} = (\sum_{t=1}^T |x_{t-1}|/\sqrt{T})\check{\beta}$ as in \eqref{numCauchy}, and
\[
\hat{\omega}^2 = \frac{1}{T}\sum_{t=1}^T \hat{u}_t^2,
\qquad
\hat{u}_t = y_t - \hat{\beta}x_{t-1}.
\]
Here, $\hat{\omega}^2$ estimates $\omega^2=\int_0^1\sigma^2(r)\,dr$ using OLS residuals. As noted by \citet{shephard2020}, the Cauchy-based variance estimator performs poorly because the Cauchy estimator converges more slowly and less efficiently than OLS when $(x_t)$ is heavy-tailed or nearly integrated. Hence, we use OLS residuals to improve efficiency.\footnote{A related approach is employed by \citet{KMS2015} in the IVX framework of \citet{PM}.}

We assume:

\begin{ass}\label{assumption-ols}
$\sum_{t=1}^T x_{t-1}u_t = o_p\!\left(\sqrt{T\sum_{t=1}^T x_{t-1}^2}\right)$.
\end{ass}

Assumption~\ref{assumption-ols} is very general and holds in many time-series settings. It is weaker than Assumption~3.2 of \citet{ibragimov2023new}, which requires $\sum_{t=1}^T x_{t-1}u_t = O_p\!\left(\sqrt{T^p\sum_{t=1}^T x_{t-1}^2}\right)$ for $p\in[0,1/16)$. As shown in \citet{ibragimov2023new}, this holds with $p=0$ when $(x_t)$ is either (near) unit root or stationary with finite variance; it also applies to certain stationary heavy-tailed processes (see, e.g., \citeauthor{Samorodnitsky2007}, \citeyear{Samorodnitsky2007}).

Under Assumption~\ref{assumption-ols},
\[
|\hat{\beta}-\beta| = o_p\!\left(\sqrt{T\big/\sum_{t=1}^T x_{t-1}^2}\right),
\quad\text{and hence}\quad \hat{\omega}^2\to_p\omega^2.
\]
The asymptotic properties of $\tau(\check{\beta})$ follow.

\begin{prop}\label{proposition-2-2}
Let Assumptions~\ref{assumption-mds}, \ref{assumption-1-2}, \ref{assumption-2-1}, and \ref{assumption-ols} hold.

(a) Under $H_0:\beta=0$,
\[
\tau(\check{\beta})\to_d\mathbb{N}(0,1).
\]

(b) Under $H_A:\beta\neq0$,
\[
\tau(\check{\beta})
= \beta\frac{\sum_{t=1}^T |x_{t-1}|}{\omega\sqrt{T}}(1+o_p(1)) + O_p(1),
\]
so $|\tau(\check{\beta})|\to_p\infty$ whenever $\sum_{t=1}^T |x_{t-1}|/\sqrt{T}\to_p\infty$.
\end{prop}

The conclusions of Proposition~\ref{proposition-2-2} remain valid under weaker conditions. For instance, if Assumptions~\ref{assumption-mds} and~\ref{assumption-ols} hold and
\[
\frac{1}{T}\sum_{t=1}^T v_t^2\to_p\omega^2>0,
\qquad
\frac{1}{\sqrt{T}}\sum_{t=1}^T \text{sign}(x_{t-1})u_t\to_d\omega Z,
\]
where $Z\sim\mathbb{N}(0,1)$ is independent of $\omega^2$, then $\tau(\check{\beta})$ retains its asymptotic validity. These conditions include stationary volatility with $E[v_t^2]=\omega^2$. Hence, Assumptions~\ref{assumption-1-2} and~\ref{assumption-2-1} can be interpreted as primitive sufficient conditions accommodating persistent volatility in predictive regression data.\smallskip

\noindent \textbf{Remark.} Proposition~\ref{proposition-2-2}(a) also holds if $\tau(\check{\beta})$ uses $\bar{\omega}^2=T^{-1}\sum_{t=1}^T y_t^2$ instead of $\hat{\omega}^2$, since $\beta=0$ under $H_0$. Moreover, the corresponding test remains consistent when $(x_t)$ is stationary with finite variance or follows a generalized local-to-unity process \citep{dou2021generalized}. However, it can be inconsistent for heavy-tailed $(x_t)$. For instance, if $(x_t)$ is i.i.d.\ $\alpha$-stable with $\alpha\in(0,2)$ and independent of $(u_t)$, then
\[
\bar{\omega}^2 = \beta^2 \left(\frac{1}{T}\sum_{t=1}^T x_{t-1}^2\right)(1+o_p(1)),
\quad
\tau(\check{\beta}) =
\frac{\sum_{t=1}^T |x_{t-1}|}{\sqrt{\sum_{t=1}^T x_{t-1}^2}}(1+o_p(1))
= O_p(1),
\]
by the generalized central limit theorem (see \citeay{feller1971}; \citeay{logan1973}; \citeay{Davis1983}; \citeay{DavisResnick1986}). Thus, the use of $\hat{\omega}^2$ (or another consistent estimator under both $H_0$ and $H_A$) is crucial for the consistency of the hybrid test.

\section{Extensions}

This section extends the inference methods developed in Section~3 to models with multiple predictors and to regressions including an intercept. The extension makes explicit a recurring trade-off in the predictive-regression literature. Methods such as IVX are designed to deliver flexible multivariate and conditional inference under persistent endogenous regressors, whereas the Cauchy-based procedures emphasize robustness to persistence, heavy tails, and persistent volatility. Our goal is therefore not to design the most efficient multivariate procedure under strong structure, but to provide simple and robust inference methods that rely on minimal assumptions on the predictors and volatility processes.

\subsection{Predictive Regressions with Multiple Predictors}

Consider a predictive regression with $K$ predictors $x_t = [x_{1,t}, \ldots, x_{K,t}]'$:
\begin{align}
y_t &= x_{t-1}' B + u_t \notag \\
    &= \beta_{1,K} x_{1,t-1} + \cdots + \beta_{K,K} x_{K,t-1} + u_t, 
    \qquad B = [\beta_{1,K}, \ldots, \beta_{K,K}]'. \label{PredRegM}
\end{align}
The objective is to test the joint predictability of the covariates, that is,
\[
H_0: \beta_{1,K} = \cdots = \beta_{K,K} = 0.
\]

We construct a testing procedure for $H_0$ based on the univariate inference methods in Section~3. Specifically, we estimate $K$ univariate predictive regressions
\[
y_t = \beta_k x_{k,t-1} + u_{k,t}, \qquad k = 1, \ldots, K,
\]
and test each null hypothesis
\[
H_0^{(k)}: \beta_k = 0, \qquad k = 1, \ldots, K.
\]
Clearly, $H_0$ implies $H_0^{(k)}$ for all $k$. The converse also holds under mild regularity conditions, as shown below.

\begin{lem}\label{lemma-4-1}
Consider model~\eqref{PredRegM} and define 
$z_t = [\operatorname{sign}(x_{1,t}), \ldots, \operatorname{sign}(x_{K,t})]'$. 
Suppose that for each $t \in \{1,\dots,T\}$,
$ E[z_{t-1} u_t] = \mathbf{0}_{K\times1}$, 
$0 <  E[|x_{k,t-1}|] < \infty$ for all $k$, and that the matrix 
$ E[z_{t-1}x_{t-1}']$ is invertible.\footnote{%
Even when $\operatorname{sign}(x_{t-1})$ is constant, the univariate Cauchy estimator and associated tests remain well defined. In this case, the estimator simplifies to 
$\check{\beta} = \sum_{t=1}^T y_t \big/ \sum_{t=1}^T x_{t-1}$, 
implying $\check{\beta}-\beta = \sum_{t=1}^T u_t \big/ \sum_{t=1}^T x_{t-1}$, and Proposition~\ref{proposition-2-2} continues to hold. 
Hence, the lack of sign variation does not affect the validity of our methods in the univariate case. 
In the multiple-predictor case, however, the invertibility of $ E[z_{t-1} x_{t-1}']$ imposes mild restrictions on sign changes across predictors. 
For practical applications, one can induce variation in the sign instrument by recentering the predictor, for example, 
$\tilde{x}_{t-1} = x_{t-1} - t^{-1}\sum_{s=1}^t x_{s-1}$, 
which preserves both the martingale structure and the asymptotic validity of the estimator.%
} 
Then, the joint null hypothesis $H_0$ holds if and only if $H_0^{(k)}$ holds for all $k = 1, \ldots, K$.
\end{lem}

Lemma~\ref{lemma-4-1} justifies the use of multiple hypothesis testing based on univariate Cauchy estimators.\footnote{See \citet{Harvey2015} for an application of the multiple-testing framework in predictive regressions, and \citet{KMS2015} for joint-predictability tests in the IVX framework. Note that the IVX approach may lose validity under heavy-tailed predictors or continuous-time data, whereas our method remains robust in such settings.} 
In conjunction with the hybrid test introduced in Section~3.2, we compute the statistic $\tau(\check{\beta}_k)$ for each parameter $\beta_k$, where $\check{\beta}_k$ denotes the corresponding Cauchy estimator. 
Let $p_k$ denote its $p$-value. 
The joint null hypothesis $H_0$ is rejected at level $\alpha$ if $\min_k p_k \le \alpha / K$, following the Bonferroni correction.

This approach directly extends the univariate robust inference procedure to a multivariate setting and requires only mild conditions for the equivalence between $H_0$ and $\{H_0^{(k)}\}_{k=1}^K$. 
The Bonferroni correction imposes no assumptions on the joint distribution of the test statistics, which motivates its use here (see \citeay{Holm1979}; \citeay{BenjaminiHochberg1995}; \citeay{Shaffer1995}). This feature is particularly useful in the present setting because the dependence among component Cauchy statistics can be complicated when regressors are persistent, endogenous, or heavy-tailed. The cost is familiar: Bonferroni inference can be conservative and does not by itself provide conditional inference on an individual coefficient after controlling for the remaining predictors. This limitation of Bonferroni-type methods in multivariate predictive regressions is also emphasized by \citet{Phillips2015}.

We also note that under some additional regularity conditions, the joint hypothesis can be tested directly using a Wald-type statistic:
\[
W = 
\left( \sum_{t=1}^T z_{t-1} y_t \right)' 
\left( \hat{\omega}^2 \sum_{t=1}^T z_{t-1} z_{t-1}' \right)^{-1}
\left( \sum_{t=1}^T z_{t-1} y_t \right).
\]
Under $H_0$, the normalized score $T^{-1/2}\sum_{t=1}^T z_{t-1}y_t$ is in general \emph{mixed} normal, because the predictable covariation of the sign processes across regressors has a random limit when the regressors are persistent. Provided the self-normalizing matrix $\hat{\omega}^2\,\frac1T\sum_{t=1}^T z_{t-1}z_{t-1}'$ consistently estimates this (random) covariance and the limiting Gaussian is independent of it, the studentized statistic is asymptotically pivotal and $W \to_d \chi_K^2$.\footnote{As mentioned earlier, $z_t$ can be interpreted as an instrument. Therefore, one may use an alternative instrument, as in \citet{shephard2020}, and construct a Wald-type test accordingly.} The latter independence holds automatically for stationary regressors, but for persistent regressors it requires the limiting sign structure to be asymptotically independent of the innovations -- a condition that the univariate component-wise tests do not need. We provide a detailed analysis, including these conditions and the resulting robustness trade-off, in Online Appendix~A.1.

While both the Wald test and the Bonferroni-adjusted marginal tests evaluate the same null hypothesis, they are not equivalent procedures except under special cases such as asymptotic orthogonality of the regressors. The Wald test exploits the full covariance structure and yields an elliptical rejection region, whereas the Bonferroni procedure yields a rectangular region and controls family-wise error. Consequently, the procedures may differ in finite-sample and asymptotic power when regressors are correlated.\footnote{Again, our goal is not to design efficient procedures but to provide simple and robust methods that rely on minimal assumptions. We leave a systematic comparison between the Bonferroni-type multiple-testing procedure and the Wald-type joint test for future research.} The two procedures are thus complementary. Because the Bonferroni procedure makes no assumption on the joint dependence of the regressors, it preserves the full robustness of the univariate tests -- to endogenous, persistent, and heavy-tailed predictors and to persistent volatility -- and is our recommended default for joint predictability testing in these settings. The Wald test instead exploits the joint covariance structure \textcolor{black}{to gain efficiency, and is the preferable option} when the regressors are stationary or asymptotically exogenous; for strongly persistent endogenous regressors its validity requires the additional condition noted above (see also the simulation evidence in Online Appendix~B.3). In that leading case, the IVX-based joint and conditional tests of \citet{KMS2015}, which are designed for multivariate predictive regressions with persistent and endogenous regressors, provide an alternative, albeit one that is not robust to heavy tails or persistent volatility. \textcolor{black}{In its present form, our Wald test also covers only the joint null; extending it to conditional inference on individual coefficients or subsets is left for future research (see Online Appendix~A.1).}


\subsection{Predictive Regressions with an Intercept}

The analysis in Section~3 assumes that the intercept $\alpha=0$ in~\eqref{PredRegr1}. When $\alpha \neq 0$, it must be properly accounted for. A natural starting point is the demeaned model
\begin{align}
y_t - \bar{y}_T = \beta (x_{t-1} - \bar{x}_T) + u_t - \bar{u}_T, 
\qquad t = 1, \ldots, T, \label{eq-demean}
\end{align}
where $\bar{z}_s = s^{-1}\sum_{t=1}^s z_t$ for $z_t\in\{y_t, x_{t-1}, u_t\}$. 
However, $(u_t - \bar{u}_T)$ is not a martingale difference sequence (MDS) with respect to $(\mathcal{F}_t)$, invalidating the martingale CLT used in Sections~2 and~3. Specifically, the Cauchy estimator becomes
\[
\check{\beta} - \beta 
= \Biggl( \sum_{t=1}^T |x_{t-1} - \bar{x}_T| \Biggr)^{-1}
\sum_{t=1}^T \text{sign}(x_{t-1} - \bar{x}_T)(u_t - \bar{u}_T),
\]
which is problematic because: (i) $u_t - \bar{u}_T$ is not an MDS, and (ii) $\text{sign}(x_{t-1} - \bar{x}_T)$ is not $\mathcal{F}_{t-1}$-measurable. Thus, the theory in Section~3 is not directly applicable.\footnote{Recursive demeaning using $\bar{y}_t$ instead of $\bar{y}_T$ does not resolve this issue since $u_t - \bar{u}_t$ is not an MDS either.}

To restore the MDS property, we instead difference the model:
\[
y_t - y_{t-1} = \beta (x_{t-1} - x_{t-2}) + (u_t - u_{t-1}),
\]
and estimate this first-differenced (FD) model on alternating subsets of observations.\footnote{Differencing $y_t$ to eliminate the nonzero mean is a widely used technique in the analysis of dynamic panel data with individual heterogeneity. See, among others, \cite{AndersonHsiao1981}; \cite{AndersonHsiao1982}; \cite{ArellanoBond1991}.} We focus on the even-indexed observations and define the modified Cauchy estimator:
\[
\check{\beta}_e 
= (D_T^e)^{-1} \sum_{t=2}^{T/2} \text{sign}(x_{2t-2})(y_{2t} - y_{2t-1}),
\quad
D_T^e = \sum_{t=2}^{T/2} \text{sign}(x_{2t-2})(x_{2t-1} - x_{2t-2}).
\]

This estimator has two key properties.  
First, for even-indexed data, the regression error $u_t^e = u_{2t} - u_{2t-1}$ forms an MDS with respect to $\mathcal{F}_t^e := \mathcal{F}_{2t}$ for $t=1,\ldots,T/2$.\footnote{For odd-indexed data, $u_t^o = u_{2t+1} - u_{2t}$ forms an MDS with respect to $\mathcal{F}_t^o := \mathcal{F}_{2t+1}$, yielding an analogous estimator $\check{\beta}_o$.}  
Second, $\check{\beta}_e$ can again be viewed as an IV estimator, but it uses $\text{sign}(x_{t-2})$, which is $\mathcal{F}_{t-2}$-measurable, as the instrument.\footnote{More generally, one may use $\text{sign}\!\left(\sum_{l\le2}w_l x_{t-l}\right)$ for deterministic weights $\{w_l\}$, provided $ E[\text{sign}(\sum_{l\le2}w_l x_{t-l})(x_{t-1}-x_{t-2})]\neq0$.}  
Hence, $\text{sign}(x_{2t-2})(u_{2t}-u_{2t-1})$ is an MDS with respect to $(\mathcal{F}_t^e)$.

The inference procedures of Section~3 remain valid for $\check{\beta}_e$. In particular, the hybrid test in Section~3.2 can be implemented as
\begin{equation}\label{hybrid_constant}
\tau(\check{\beta}_e) = \frac{\check{\gamma}_e}{\hat{\omega}},
\end{equation}
where $\check{\gamma}_e = D_T^e \check{\beta}_e / \sqrt{T/2}$, and
\[
\hat{\omega}^2 = \frac{1}{T}\sum_{t=1}^T \hat{u}_t^2, 
\qquad
\hat{u}_t = (y_t - \bar{y}_T) - \hat{\beta}(x_{t-1}-\bar{x}_T),
\]
with $\hat{\beta}$ the OLS estimator from the demeaned model~\eqref{eq-demean}. Note that $\hat{\omega}^2$ is based on the full sample, whereas $\check{\beta}_e$ uses only even-indexed data. The asymptotic validity of this hybrid procedure is established next.

\begin{cor}\label{corollary-3-2}
Let Assumptions~\ref{assumption-mds}, \ref{assumption-1-2}, and \ref{assumption-2-1} hold, and suppose Assumption~\ref{assumption-ols} holds with $x_{t-1}$ replaced by $x_{t-1}-\bar{x}_T$.

(a) Under $H_0:\beta=0$,
\[
\tau(\check{\beta}_e)\to_d\mathbb{N}(0,1).
\]

(b) Under $H_A:\beta\neq0$,
\[
\tau(\check{\beta}_e)
= \beta\,\frac{\sum_{t=1}^{T/2}\text{sign}(x_{2t-2})(x_{2t-1}-x_{2t-2})}
{\omega\sqrt{T/2}}(1+o_p(1)) + O_p(1),
\]
so that $|\tau(\check{\beta}_e)|\to_p\infty$ whenever 
$\bigl|\sum_{t=1}^{T/2}\text{sign}(x_{2t-2})(x_{2t-1}-x_{2t-2})\bigr|/\sqrt{T/2}\to_p\infty$.
\end{cor}

Although the odd-indexed estimator $\check{\beta}_o$ has analogous properties, $\check{\beta}_e$ and $\check{\beta}_o$ are typically dependent, with the dependence structure determined by the DGP of $(x_t)$. Hence, unless additional assumptions are imposed, we restrict attention to a single subset of observations—either with even or odd indices.\footnote{Using only half of the data is not uncommon in predictive regressions. See, for example, \citet{ZhuCaiPeng2014} and \citet{Liu2019}, who employ long-lag differencing to eliminate intercepts. In addition, \citet{Dufour} uses a split-sample approach to address inference problems under a Markovian structure.}

Consistency of the hybrid test with an intercept requires
\[
\frac{1}{\sqrt{T/2}}\sum_{t=1}^{T/2}\operatorname{sign}(x_{2t-2})(x_{2t-1}-x_{2t-2})\to_p\infty.
\]
This holds for most stationary processes $(x_t)$ if
\[
 E[\operatorname{sign}(x_{t-1})x_t] \neq  E[|x_{t-1}|].
\]
The condition may fail for certain unit-root processes. For instance, for a random walk $x_t=x_{t-1}+\varepsilon_t^x$, it does not hold. More generally, in the local-to-unity model of \citet{Phillips_Magdalinos_2007book},
\begin{align}\label{localtounity}
x_t = \left(1+\frac{c}{T^\delta}\right)x_{t-1} + \varepsilon_t^x, 
\qquad c<0, \quad \delta\in[0,1],
\end{align}
with $\varepsilon_t^x$ satisfying Assumption~LP therein, the consistency condition becomes
\begin{align*}
\frac{1}{T^{1/2}}\sum_{t=1}^{T/2}\operatorname{sign}(x_{2t-2})(x_{2t-1}-x_{2t-2})
&= \frac{c}{T^{1/2+\delta}}\sum_{t=1}^{T/2}|x_{2t-2}|
  + \frac{1}{T^{1/2}}\sum_{t=1}^{T/2}\operatorname{sign}(x_{2t-2})\varepsilon_{2t-1}^x.
\end{align*}
The first term diverges if and only if $\delta < 1$ (see Lemma~3.2 of \citeauthor{Phillips_Magdalinos_2007book}, \citeyear{Phillips_Magdalinos_2007book}). Since $\varepsilon_t^x$ and $x_{t-1}$ may be dependent, it is generally the case that $ E[\operatorname{sign}(x_{2t-2})\varepsilon_{2t-1}^x] \neq 0$, implying that the second term also diverges. Consequently, inconsistency arises only when $\delta = 1$ and $ E[\operatorname{sign}(x_{2t-2})\varepsilon_{2t-1}^x] = 0$. In all other cases (i.e., $0 \le \delta < 1$ or when the expectation is nonzero), the test remains consistent.

Accommodating a nonzero intercept thus requires first differencing, and this is not without cost. Since the inclusion of an intercept is standard practice in empirical predictive regressions in economics and finance, this is the practically relevant case. First differencing reduces power: the test based on $\tau(\check{\beta}_e)$ is less powerful than that based on $\tau(\check{\beta})$ when the intercept is zero, a comparison made precise in Online Appendix~A.2. For local-to-unity regressors with persistence exponent $\delta$, the first-differenced statistic diverges at rate $T^{1/2-\delta/2}$ rather than $T^{1/2+\delta/2}$, so the power loss grows with the persistence of the regressor. Moreover, in the multivariate setting this cost compounds with the conditions required by the Wald-type joint test of Section~4.1: combining first differencing with the Wald statistic (see Online Appendix~A.1) requires both consistent estimation of the joint covariance and, for strongly persistent regressors, an asymptotic exogeneity condition. We stress, however, that \emph{validity} -- size control and consistency -- is retained with an intercept through the first-differenced tests and their Bonferroni combination; what is reduced is power. By contrast, the IVX-based methods of \citet{KMS2015} and \citet{KMS2023} are specifically designed for multivariate predictive regressions with intercepts under persistence and endogeneity, and provide a comprehensive framework for flexible joint inference that our robustness-oriented procedures complement.

\section{Finite Sample Performance}

This section investigates the finite-sample performance of the proposed methods.  
Two sets of simulation experiments are conducted.  
The first set is based on a continuous-time model and compares our robust $t$-statistic–based tests, $t_q(\check{\gamma})$ for $q \in \{8,12,16\}$, and the hybrid test $\tau(\check{\beta})$ with the Cauchy RT test of \citet{CJP2016} and the Cauchy VC test of \citet{ibragimov2023new}. In this setting, the intercept is suppressed in the test statistics. The second set is based on a discrete-time predictive regression model and compares our procedures, $\tau(\check{\beta}_e)$ and $\tau(\check{\beta}_o)$, with the IVX test of \citet{KMS2015}. In this setting, the statistics are computed with the intercept included in the estimation.\footnote{The number of replications is 10,000 for each setting.}  

\subsection{Continuous-Time Experiments}

\subsubsection{Simulation Design}

Following \citet{CJP2016} and \citet{ibragimov2023new}, we consider the continuous-time predictive regression model
\begin{align}
dY_t &= \beta X_t\,dt + dU_t, \label{MC1}\\
dX_t &= -\frac{\bar{\kappa}}{T} X_t\,dt + \sigma_t\,dV_t, 
\qquad 
dU_t = \sigma_t\left(dW_t + \int_{\mathbb{R}} x\,\Lambda(dt,dx)\right), \nonumber
\end{align}
where $V_t$ and $W_t$ are Brownian motions with $ E(V_t W_t) = -0.98t$.  
The constant term in the predictive regression is set to zero, and recursive demeaning is applied as \citet{CJP2016}.  
The model is observed at interval $\Delta = 1/252$, corresponding to daily observations, so that a sample of length $T$ years contains $252T$ observations.

The volatility process $\sigma_t$ follows one of the following specifications:
\begin{itemize}
    \item \textbf{CNST (Constant Volatility):} $\sigma_t^2 = \sigma_0^2$, with $\sigma_0 = 1$.
    \item \textbf{SB (Structural Break):} $\sigma_t = \sigma_0 + (\sigma_1 - \sigma_0)1\{t/T \ge 4/5\}$, with $\sigma_0 = 1$ and $\sigma_1 = 4$.
    \item \textbf{RS (Regime Switching):} $\sigma_t = \sigma_0(1 - s_t) + \sigma_1 s_t$, 
    where $s_t$ is a two-state Markov chain independent of $(Y_t, X_t)$, 
    with transition matrix
    \[
    P_t =
    \begin{pmatrix}
    0.8 & 0.2 \\ 0.8 & 0.2
    \end{pmatrix}
    +
    \begin{pmatrix}
    0.2 & -0.2 \\ -0.8 & 0.8
    \end{pmatrix}
    \exp\!\left(-\frac{\bar{\lambda}}{T}t\right),
    \]
    initialized at its invariant distribution, where $\bar{\lambda} = 60$, $\sigma_0 = 1$, and $\sigma_1 = 4$.
    \item \textbf{GBM (Geometric Brownian Motion):} 
    $d\sigma_t^2 = \frac{1}{2}\frac{\bar{\omega}^2}{T}\sigma_t^2\,dt 
    + \frac{\bar{\omega}^2}{\sqrt{T}}\sigma_t^2\,dZ_t$, 
    where $Z_t$ is a Brownian motion correlated with $W_t$ such that $ E(W_t Z_t) = -0.4t$ and $\bar{\omega} = 9$.
\end{itemize}

Following \cite{CJP2016}, we set $T \in \{20, 50, 100\}$ \textcolor{black}{(i.e., $252T$ daily observations)} and $\bar{\kappa} \in \{0, 5, \textcolor{black}{20}\}$ for the persistence parameter in~\eqref{MC1}. We consider a two-sided test of $H_0: \beta = 0$ against $H_A: \beta \neq 0$.

\subsubsection{Results}

We first assess the empirical size of each test under $H_0: \beta = 0$.  
The results for the four volatility models (CNST, SB, RS, and GBM) are reported in Table~\ref{tab1}.  Overall, both the $t$-statistic–based tests and the hybrid method exhibit satisfactory size performance, closely matching the nominal levels and performing comparably to the Cauchy RT and Cauchy VC tests.  
Among the $t$-based procedures, moderate partition numbers ($q = 12$ or $16$) provide the most stable results, whereas smaller $q$ values tend to be mildly undersized.  
In the GBM case, where volatility is endogenously persistent, the $t$-statistic–based tests become slightly conservative but remain competitive with the Cauchy RT and VC methods.

Next, we analyze the finite-sample power properties of the tests.  
We consider $\beta \in \{0.004 k, k=1,\cdots,5\}$ under the same volatility specifications.  
The results are summarized in Tables~\ref{tab2}–\ref{tab5}.  The proposed tests exhibit power comparable to that of the Cauchy RT and Cauchy VC procedures.  
For small samples ($T = 20$), the Cauchy RT and VC tests occasionally show higher power, but the difference diminishes as $T$ increases.  
In certain settings, our methods even outperform the existing approaches.  
For instance, $t_{16}(\check{\gamma})$ dominates under $\beta = 0.02$, $\bar{\kappa} = \textcolor{black}{5}$, and regime switching volatility (Table~\ref{tab4}), whereas the hybrid test $\tau(\check{\beta})$ performs best under $\beta = 0.004$, $\bar{\kappa} = 20$, $T = 20$, and regime switching volatility.

In summary, all four robust inference procedures—Cauchy RT, Cauchy VC, $t_q(\check{\gamma})$, and $\tau(\check{\beta})$—deliver accurate size control and strong discriminatory power under endogenously persistent regressors and persistent volatility.
While the Cauchy RT requires high-frequency data and a time transformation, and the Cauchy VC involves nonparametric volatility filtering with a tuning parameter, our proposed $t$-statistic and hybrid methods are much simpler to implement and require neither.
Hence, these approaches are best viewed as complementary: the Cauchy RT and Cauchy VC are preferable in high-frequency environments, whereas our procedures provide robust and easily implementable alternatives in more general settings.
It is also worth emphasizing that the proposed methods, like the Cauchy VC, are applicable to both continuous- and discrete-time models, whereas the Cauchy RT method is restricted to the continuous-time framework.

\subsection{Discrete-Time Experiments}

\subsubsection{Simulation Design}

We now examine the finite-sample performance of the proposed tests in a discrete-time setting with an intercept, comparing them to the IVX test of \citet{KMS2015}.  
The DGP is specified as
\begin{align}
y_t &= \beta x_{t-1} + \sigma_t \varepsilon_t, \label{MC2} \\
x_t &= \left(1 - \frac{\bar{\kappa}}{T}\right) x_{t-1} + \sigma_t \eta_t, \nonumber
\end{align}
for $t = 2, \ldots, 12T$, where $T \in \{20, 50, 100\}$ corresponds to 20, 50, and 100 years of monthly data (i.e., 240, 600, and 1200 monthly observations, respectively). The constant term in the predictive regression is set to zero.
We set $\beta \in \{0.5k : k = 0, 1, \ldots, 5\}$ and $\bar{\kappa} \in \{0, 50, 100\}$, and consider a one-sided test of $H_0\!: \beta = 0$ against $H_A\!: \beta > 0$.\footnote{The IVX test of \citet{KMS2015} performs well in two-sided testing for a broad class of models with constant volatility. However, as shown in \citet{DEMETRESCU2023105271}, the IVX method exhibits severe size distortions in one-sided tests when regressors are highly persistent and endogenous. For this reason, we focus on the one-sided case to highlight the performance of our methods in this setting. Simulations for two-sided and joint hypothesis tests are also reported in Online Appendix~B. Notably, while the IVX approach demonstrates reliable size control under constant volatility, it exhibits significant size distortions in the presence of non-constant volatility. In contrast, our methods remain robust to such heteroskedasticity.}

The innovation process $\eta_t$ follows an MA($q$) process:
\begin{align}
\eta_t = \sum_{j=0}^q C_j v_{t-j},\label{MC3}
\end{align}
where $(\varepsilon_t, v_t)$ are jointly normal with correlation $-0.98$. In our simulations, we consider an MA(1) process with $C_0 = C_1 = 1/\sqrt{2}$.\footnote{Additional simulation results with an MA(3) process are reported in Online Appendix~B.1. Moreover, as discussed in Section~4.2, the intercept-robust version of our test, $\tau(\check{\beta}_e)$, becomes inconsistent when $x_t$ follows the local-to-unity process \eqref{localtounity} with $\delta = 1$ and $\eta_t$ is serially independent. For this reason, we do not consider i.i.d.\ $\eta_t$ in our simulations. See Online Appendix~A.2, for a detailed theoretical power analysis of the test statistic $\tau(\check{\beta}_e)$.} The volatility process $\sigma_t$ follows the same specifications as in the continuous-time simulations, except that the GBM model is excluded.

We implement the hybrid tests based on the even and odd observations, denoted by $\tau(\check{\beta}_e)$ and $\tau(\check{\beta}_o)$, respectively (see \eqref{hybrid_constant}). For comparison, we also include the IVX test of \citet{KMS2015}, in which the intercept is included in the estimation.

\subsubsection{Results}

The results, summarized in Tables~\ref{tab6}–\ref{tab10}, indicate that the proposed tests exhibit excellent size control under the null hypothesis across all DGPs, whereas the IVX test is substantially oversized, particularly when volatility is nonstationary or exhibits structural breaks. Furthermore, the hybrid approaches demonstrate nontrivial power, even though they are constructed using only half of the observations. 

Overall, these findings corroborate the theoretical robustness of our methods. They remain valid under heavy-tailed, endogenous, and persistent regressors, as well as under heteroskedastic and persistent volatility. In contrast, the IVX test exhibits severe size distortions in one-sided tests across all cases, including the constant volatility setting, as documented by \citet{DEMETRESCU2023105271}. 

We also conduct two-sided tests and report the results in Online Appendix~B.2. In this case, the IVX method exhibits good size performance under constant volatility but suffers from severe size distortions in settings with time-varying volatility, whereas our methods maintain accurate size control across all specifications.

Taken together, these results suggest that the proposed robust procedures provide a practical and reliable alternative to existing inference methods for predictive regressions in both continuous- and discrete-time frameworks.

\section{Empirical Application}

To illustrate the empirical performance of the proposed tests relative to the Cauchy RT, Cauchy VC, and IVX tests, we reexamine the dataset used by \citet{CJP2016} to test the predictability of stock returns using the dividend--price (D/P) and earnings--price (E/P) ratios as predictors.\footnote{Further empirical results, including intercept-robust estimations and Wald-type tests compared to IVX \citep{KMS2015}, are detailed in Online Appendix~C.} For stock returns, we employ the NYSE/AMEX value-weighted index and the S\&P~500 index obtained from the Center for Research in Security Prices (CRSP). The dividend--price ratio is defined as the annual dividend divided by the current total market value. Further details on data construction are provided in Section~6.1 of \citet{CJP2016}.  

Following \citet{CJP2016}, we estimate two types of predictive regressions: one based on all returns and another based only on returns generated from the diffusive component of stock prices, obtained by first testing for jumps and removing observations corresponding to detected jumps. In all cases, we apply one-sided tests.  

The results are reported in Table~\ref{tab12}. As shown in Panels~C and~D, none of the tests reject the null hypothesis of unpredictability for the S\&P~500 data when the E/P ratio is used as a predictor. By contrast, when the D/P ratio serves as a predictor, the proposed tests -- $t_q(\check{\gamma})$ with $q=12,16$ and $\tau(\check{\beta})$ -- reject the null of unpredictability for several cases: CRSP (yearly without jump removal; quarterly with jump removal) and S\&P~500 \textcolor{black}{(yearly, with and without jump removal)}. In contrast, the Cauchy RT test fails to reject the null in all cases, while the Cauchy VC test yields qualitatively similar conclusions to our proposed tests, except that it additionally rejects the null for \textcolor{black}{S\&P~500 (quarterly without jump removal; monthly with jump removal) but does not reject for S\&P~500 (yearly with jump removal)}. The IVX statistic of \citet{KMS2015} provides a useful benchmark because it is designed for predictive regressions with persistent and endogenous predictors. In Table~\ref{tab12}, IVX gives stronger evidence of predictability for the D/P ratio, rejecting in almost all D/P specifications, including several monthly and quarterly cases in which the Cauchy-based robust tests do not reject.\footnote{This comparison should be interpreted in light of the simulation evidence in Online Appendix~B: IVX performs well under standard conditions, but can exhibit size distortions under time-varying volatility, including structural breaks and regime switching. Because stock-return volatility is itself persistent and time-varying, the stronger IVX rejections for the D/P ratio should be read with this in mind, as they may in part reflect such distortions rather than additional genuine predictability. This is one motivation for reporting the robustness-oriented Cauchy procedures alongside IVX rather than treating the methods as substitutes.} At the same time, IVX does not reject for the E/P ratio, so its broad empirical conclusion is consistent with ours: the evidence for return predictability is concentrated in the dividend--price ratio rather than the earnings--price ratio. The differences across methods are therefore mainly about the strength and frequency of D/P rejections, reflecting the different design goals of IVX and the proposed robustness-oriented Cauchy procedures.  

Consistent with our simulation evidence, the Cauchy RT test demonstrates strong finite-sample power but requires high-frequency data due to its reliance on a continuous-time approximation.\footnote{For the Cauchy RT test in our simulations, we estimate the discretized time-changed regression using $n$ lower-frequency observations, with $n = 12T$ (approximately monthly), generated by the random time-sampling scheme described in Section~5 of \citet{CJP2016}.} The mixed empirical results -- where the Cauchy RT test fails to reject the null while both the proposed methods and the Cauchy VC test do reject -- may reflect the limited accuracy of the continuous-time approximation when applied to monthly, quarterly, or yearly data. Evaluating the robustness of the continuous-time approximation underlying the Cauchy RT test remains an interesting topic for future research.

\section{Conclusion}

This paper introduces two robust inference methods for predictive regressions, addressing key econometric challenges commonly encountered in empirical finance, such as endogenously persistent or heavy-tailed regressors and persistent volatility in errors. Building on the Cauchy estimation framework, we develop two simple yet theoretically rigorous procedures: a $t$-statistic–based approach and a hybrid method. Both methods are computationally straightforward and applicable to continuous- and discrete-time models alike. \textcolor{black}{The methods target robust inference on individual predictive coefficients. With multiple predictors, they yield a conservative Bonferroni-based joint test, whereas the Wald-type extension requires stronger assumptions and does not retain their full robustness; accommodating an intercept further requires first differencing, at some cost in power.}

Simulation evidence demonstrates that the proposed tests perform well in finite samples, maintaining correct size and competitive power under a wide range of data-generating processes, including those characterized by stochastic volatility, structural breaks, and regime switching. Although our procedures require the assumption of asymptotically exogenous volatility, they exhibit excellent robustness and complement existing methods, including the IVX method of \cite{KMS2015}, the Cauchy RT test of \citet{CJP2016}, and the Cauchy VC test of \citet{ibragimov2023new}.

In an empirical application to stock return predictability, we use the dividend--price and earnings--price ratios as predictors for excess returns on major U.S.\ equity indices. The results indicate that the dividend--price ratio possesses predictive power, while the earnings--price ratio does not significantly forecast returns. Overall, the proposed inference procedures offer a practical, theoretically sound, and implementable alternative to existing methods for robust inference in predictive regressions.


\newpage

\newcolumntype{C}{>{\centering\arraybackslash}X}

\begin{table}[h!]
\begin{center}
\footnotesize
\caption{Size for Continuous-Time Models\label{tab1}}

\begin{tabularx}{0.67\textwidth}{lcCCCCCCCCC} \toprule
&&\multicolumn{3}{c}{$\bar{\kappa}=0$}&\multicolumn{3}{c}{$\bar{\kappa}=5$}&\multicolumn{3}{c}{$\bar{\kappa}=20$}\\
\cmidrule(r){3-5}\cmidrule(r){6-8}\cmidrule(r){9-11}
T&&20&50&100&20&50&100&20&50&100\\\hline
CNST&Cauchy RT&4.9&5.3&5.0&5.4&4.7&5.0&5.1&5.1&5.1\\
&Cauchy VC&5.0&5.3&4.9&5.0&5.1&5.0&5.0&4.8&5.0\\
&$t_8(\check{\gamma})$&4.6&5.2&4.9&5.0&5.1&5.4&4.9&5.0&5.2\\
&$t_{12}(\check{\gamma})$&4.7&5.0&5.1&5.1&4.9&5.4&5.2&4.9&5.0\\
&$t_{16}(\check{\gamma})$&4.8&5.1&5.0&5.0&4.9&5.3&4.9&4.9&5.1\\
&$\tau(\check{\beta})$&4.5&5.2&5.0&4.8&5.1&5.2&4.9&4.8&5.2\\\hline
SB&Cauchy RT&5.0&5.1&5.8&5.3&5.0&5.9&5.0&4.9&5.9\\
&Cauchy VC&6.7&6.3&4.7&6.5&6.0&4.9&6.4&6.0&5.0\\
&$t_8(\check{\gamma})$&3.7&3.9&3.6&4.1&3.7&4.1&3.7&3.7&3.9\\
&$t_{12}(\check{\gamma})$&4.2&4.6&4.2&4.6&4.5&4.3&4.2&4.5&4.1\\
&$t_{16}(\check{\gamma})$&4.6&4.6&4.6&4.7&4.5&4.7&4.6&4.5&4.8\\
&$\tau(\check{\beta})$&5.0&4.9&4.9&5.5&5.1&5.0&5.5&5.0&4.8\\\hline
RS&Cauchy RT&4.8&5.2&6.6&4.9&4.9&6.5&5.1&4.8&6.2\\
&Cauchy VC&5.4&6.1&5.1&5.1&5.8&5.2&5.8&5.8&4.7\\
&$t_8(\check{\gamma})$&4.5&5.1&4.6&4.4&4.9&5.1&5.3&4.6&4.5\\
&$t_{12}(\check{\gamma})$&4.6&5.0&4.7&4.8&4.8&5.3&4.9&4.7&5.1\\
&$t_{16}(\check{\gamma})$&4.4&4.9&4.7&4.6&4.9&5.2&5.1&4.5&5.0\\
&$\tau(\check{\beta})$&4.9&5.2&4.9&4.7&4.9&5.1&5.3&5.0&5.3\\\hline
GBM&Cauchy RT&4.7&4.4&6.6&4.5&4.4&6.5&4.5&4.5&6.4\\
&Cauchy VC&5.5&6.1&4.4&5.7&5.9&4.8&5.9&6.5&5.0\\
&$t_8(\check{\gamma})$&3.1&3.2&3.2&3.8&3.4&3.5&3.6&3.9&3.9\\
&$t_{12}(\check{\gamma})$&3.3&3.7&3.8&4.3&3.6&4.0&4.2&4.3&4.4\\
&$t_{16}(\check{\gamma})$&3.8&4.1&4.1&4.2&4.0&4.2&4.5&4.6&4.3\\
&$\tau(\check{\beta})$&4.6&5.0&4.5&4.8&4.9&4.8&4.9&5.2&5.1\\
 \bottomrule\smallskip
\end{tabularx}
\end{center}
\scriptsize{
Rejection rates of the Cauchy RT, Cauchy VC, $t_q(\check{\gamma})$ ($q \in \{8,12,16\}$), and $\tau(\check{\beta})$ tests under $\beta = 0$ (two-sided tests). Data are from the continuous-time model \eqref{MC1} with $\bar{\kappa} \in \{0,5,20\}$, sampled at \(\Delta = 1/252\) (daily); \(T \in \{20,50,100\}\) years (\(252T\) observations). CNST, SB, GBM, and RS denote constant volatility, structural breaks, geometric Brownian motion, and regime switching, respectively.}
\end{table}

\newpage

\begin{table}[h!]
\begin{center}
\footnotesize
\caption{Power for Continuous-Time Models with Constant Volatility\label{tab2}}

\begin{tabularx}{0.84\textwidth}{ccCCCCCCCCC} \toprule
&&\multicolumn{3}{c}{$\bar{\kappa}=0$}&\multicolumn{3}{c}{$\bar{\kappa}=5$}&\multicolumn{3}{c}{$\bar{\kappa}=20$}\\
\cmidrule(r){3-5}\cmidrule(r){6-8}\cmidrule(r){9-11}
T&&20&50&100&20&50&100&20&50&100\\\hline
$\beta=0.004$&Cauchy RT&8.8&25.2&85.4&6.1&9.9&22.7&6.3&13.2&27.9\\
&Cauchy VC&9.1&25.6&86.2&6.5&10.0&23.5&8.2&13.4&26.8\\
&$t_8(\check{\gamma})$&8.0&21.0&79.5&6.7&11.2&24.6&5.1&14.6&30.6\\
&$t_{12}(\check{\gamma})$&9.5&21.4&82.0&6.2&10.5&25.4&5.9&14.3&31.0\\
&$t_{16}(\check{\gamma})$&8.8&22.8&84.0&5.5&10.0&23.3&5.5&14.5&29.9\\
&$\tau(\check{\beta})$&6.4&23.5&66.6&5.3&8.6&17.5&5.9&12.9&24.3\\\hline
$\beta=0.008$&Cauchy RT&15.5&85.5&100.0&9.2&22.1&83.9&6.6&27.4&85.6\\
&Cauchy VC&17.2&86.2&100.0&8.8&23.1&84.3&11.9&27.7&84.2\\
&$t_8(\check{\gamma})$&13.8&79.9&100.0&9.1&23.5&86.7&6.0&32.5&89.2\\
&$t_{12}(\check{\gamma})$&15.0&83.4&100.0&8.9&25.5&86.8&6.5&31.7&92.0\\
&$t_{16}(\check{\gamma})$&15.0&83.7&100.0&8.3&23.7&87.0&6.1&31.9&91.7\\
&$\tau(\check{\beta})$&14.0&66.7&96.6&7.3&18.8&57.6&8.5&27.6&70.0\\\hline
$\beta=0.012$&Cauchy RT&37.3&99.2&100.0&12.2&50.1&100.0&7.7&52.2&100.0\\
&Cauchy VC&40.2&99.3&100.0&12.6&51.9&100.0&16.5&56.3&100.0\\
&$t_8(\check{\gamma})$&30.3&98.0&100.0&12.9&55.5&99.8&6.9&60.6&100.0\\
&$t_{12}(\check{\gamma})$&33.1&99.5&100.0&12.0&56.3&100.0&7.2&61.6&100.0\\
&$t_{16}(\check{\gamma})$&35.3&98.9&100.0&12.3&54.1&100.0&6.6&59.2&100.0\\
&$\tau(\check{\beta})$&32.4&87.7&99.1&9.7&34.0&93.5&12.1&48.9&99.0\\\hline
$\beta=0.016$&Cauchy RT&67.0&100.0&100.0&17.0&84.4&100.0&8.4&83.2&100.0\\
&Cauchy VC&68.8&100.0&100.0&16.6&84.7&100.0&22.4&83.3&100.0\\
&$t_8(\check{\gamma})$&58.8&99.7&100.0&18.1&88.7&100.0&7.3&88.7&100.0\\
&$t_{12}(\check{\gamma})$&60.7&100.0&100.0&16.8&86.6&100.0&7.6&89.7&100.0\\
&$t_{16}(\check{\gamma})$&62.4&99.9&100.0&16.7&86.9&100.0&7.0&89.2&100.0\\
&$\tau(\check{\beta})$&51.2&95.8&99.9&12.4&58.5&99.0&17.8&72.9&100.0\\\hline
$\beta=0.02$&Cauchy RT&86.7&100.0&100.0&23.2&98.0&100.0&9.3&98.2&100.0\\
&Cauchy VC&87.9&100.0&100.0&24.4&98.5&100.0&29.4&97.4&100.0\\
&$t_8(\check{\gamma})$&79.3&100.0&100.0&24.6&97.8&100.0&8.5&98.3&100.0\\
&$t_{12}(\check{\gamma})$&81.9&100.0&100.0&23.8&99.1&100.0&8.5&99.1&100.0\\
&$t_{16}(\check{\gamma})$&82.3&100.0&100.0&23.2&98.3&100.0&7.7&99.2&100.0\\
&$\tau(\check{\beta})$&65.2&97.5&100.0&16.3&81.3&99.6&23.8&92.3&100.0\\
 \bottomrule\smallskip
\end{tabularx}
\end{center}
\scriptsize{
Rejection rates of the Cauchy RT, Cauchy VC, $t_q(\check{\gamma})$ ($q \in \{8,12,16\}$), and $\tau(\check{\beta})$ tests under $\beta \in \{0.004, 0.008, 0.012, 0.016, 0.02\}$  (two-sided tests). Data are from the continuous-time model \eqref{MC1} with $\bar{\kappa} \in \{0,5,20\}$, sampled at \(\Delta = 1/252\) (daily); \(T \in \{20,50,100\}\) years (\(252T\) observations). The volatility is constant.}

\end{table}

\newpage

\begin{table}[h!]
\begin{center}
\footnotesize
\caption{Power for Continuous-Time Models with Structural Breaks in Volatility\label{tab3}}

\begin{tabularx}{0.84\textwidth}{ccCCCCCCCCC} \toprule
&&\multicolumn{3}{c}{$\bar{\kappa}=0$}&\multicolumn{3}{c}{$\bar{\kappa}=5$}&\multicolumn{3}{c}{$\bar{\kappa}=20$}\\
\cmidrule(r){3-5}\cmidrule(r){6-8}\cmidrule(r){9-11}
T&&20&50&100&20&50&100&20&50&100\\\hline
$\beta=0.004$&Cauchy RT&8.7&20.0&67.7&8.3&11.6&35.9&8.7&17.2&52.4\\
&Cauchy VC&8.5&20.6&57.6&8.9&10.5&19.2&8.2&14.6&25.3\\
&$t_8(\check{\gamma})$&5.6&11.3&45.6&4.4&8.7&15.7&6.4&8.4&23.3\\
&$t_{12}(\check{\gamma})$&6.6&14.8&45.3&6.7&9.1&17.5&6.2&10.5&22.7\\
&$t_{16}(\check{\gamma})$&8.2&14.6&44.1&6.9&8.9&19.5&6.6&10.1&22.9\\
&$\tau(\check{\beta})$&5.8&13.7&31.4&7.1&8.1&13.5&5.8&10.2&16.7\\\hline
$\beta=0.008$&Cauchy RT&14.4&65.1&96.6&11.7&35.7&88.4&14.5&52.2&99.6\\
&Cauchy VC&14.9&60.6&95.6&10.8&21.6&55.8&11.5&30.1&72.9\\
&$t_8(\check{\gamma})$&9.5&44.1&88.5&6.8&16.3&57.4&9.0&21.0&68.3\\
&$t_{12}(\check{\gamma})$&11.4&45.6&89.8&8.9&17.9&55.9&9.2&23.6&69.4\\
&$t_{16}(\check{\gamma})$&12.7&42.3&92.1&8.4&18.3&55.9&8.8&22.2&72.2\\
&$\tau(\check{\beta})$&8.9&32.4&81.9&8.9&14.0&34.3&7.9&16.6&49.7\\\hline
$\beta=0.012$&Cauchy RT&32.1&87.0&99.3&15.5&68.6&98.8&24.2&90.1&100.0\\
&Cauchy VC&26.1&83.1&99.2&13.5&36.0&91.1&15.5&50.3&98.8\\
&$t_8(\check{\gamma})$&16.9&72.9&97.0&9.2&32.7&89.0&12.0&43.1&95.2\\
&$t_{12}(\check{\gamma})$&20.7&74.5&99.1&11.9&34.5&90.0&12.2&41.7&98.0\\
&$t_{16}(\check{\gamma})$&20.0&74.3&99.6&10.0&31.9&90.7&12.6&41.4&98.8\\
&$\tau(\check{\beta})$&14.7&60.2&98.1&10.8&23.4&73.7&10.8&33.2&92.1\\\hline
$\beta=0.016$&Cauchy RT&50.2&95.6&99.7&22.8&89.8&100.0&37.7&99.7&100.0\\
&Cauchy VC&42.9&95.1&99.6&17.0&58.9&99.5&20.7&74.5&100.0\\
&$t_8(\check{\gamma})$&30.2&88.8&98.5&12.6&56.2&96.1&16.6&67.4&99.8\\
&$t_{12}(\check{\gamma})$&31.7&90.4&99.7&14.7&53.8&98.4&17.0&65.0&100.0\\
&$t_{16}(\check{\gamma})$&31.8&91.1&100.0&13.0&54.4&99.0&18.5&66.6&100.0\\
&$\tau(\check{\beta})$&23.2&81.4&99.8&13.0&35.6&96.9&14.3&51.8&100.0\\\hline
$\beta=0.02$&Cauchy RT&65.4&98.6&99.9&34.6&97.4&100.0&55.1&100.0&100.0\\
&Cauchy VC&58.8&98.5&100.0&21.4&78.7&99.9&27.1&91.5&100.0\\
&$t_8(\check{\gamma})$&45.5&95.3&98.8&17.6&74.0&98.2&21.7&86.2&100.0\\
&$t_{12}(\check{\gamma})$&46.1&96.4&99.8&18.0&75.7&99.7&22.7&87.5&100.0\\
&$t_{16}(\check{\gamma})$&44.7&97.9&100.0&18.0&73.9&100.0&24.9&87.9&100.0\\
&$\tau(\check{\beta})$&32.2&92.4&100.0&15.7&56.6&99.8&18.2&73.4&100.0\\
 \bottomrule\smallskip
\end{tabularx}
\end{center}
\scriptsize{
Rejection rates of the Cauchy RT, Cauchy VC, $t_q(\check{\gamma})$ ($q \in \{8,12,16\}$), and $\tau(\check{\beta})$ tests under $\beta \in \{0.004, 0.008, 0.012, 0.016, 0.02\}$  (two-sided tests). Data are from the continuous-time model \eqref{MC1} with $\bar{\kappa} \in \{0,5,20\}$, sampled at \(\Delta = 1/252\) (daily); \(T \in \{20,50,100\}\) years (\(252T\) observations). The volatility exhibits structural breaks.}
\end{table}

\newpage

\begin{table}[h!]
\begin{center}
\footnotesize
\caption{Power for Continuous-Time Models with Regime Switching in Volatility\label{tab4}}

\begin{tabularx}{0.84\textwidth}{ccCCCCCCCCC} \toprule
&&\multicolumn{3}{c}{$\bar{\kappa}=0$}&\multicolumn{3}{c}{$\bar{\kappa}=5$}&\multicolumn{3}{c}{$\bar{\kappa}=20$}\\
\cmidrule(r){3-5}\cmidrule(r){6-8}\cmidrule(r){9-11}
T&&20&50&100&20&50&100&20&50&100\\\hline
$\beta=0.004$&Cauchy RT&6.6&23.8&64.2&5.9&10.3&21.0&8.1&12.8&32.9\\
&Cauchy VC&11.0&49.9&85.9&7.8&12.9&33.0&7.4&13.0&34.5\\
&$t_8(\check{\gamma})$&7.7&22.1&77.7&6.1&11.8&22.2&7.5&10.4&28.4\\
&$t_{12}(\check{\gamma})$&8.0&23.2&83.6&6.2&10.9&22.6&7.7&11.0&28.9\\
&$t_{16}(\check{\gamma})$&7.9&22.8&84.0&6.5&11.8&22.9&7.8&10.6&30.8\\
&$\tau(\check{\beta})$&8.2&22.3&64.8&7.2&7.6&15.0&8.3&9.3&20.2\\\hline
$\beta=0.008$&Cauchy RT&11.9&69.3&96.0&8.0&20.3&65.5&12.2&32.2&90.2\\
&Cauchy VC&25.7&84.6&97.8&10.0&30.2&84.1&10.8&28.2&86.6\\
&$t_8(\check{\gamma})$&15.9&77.2&99.8&9.6&24.3&79.8&10.5&26.9&81.9\\
&$t_{12}(\check{\gamma})$&15.2&80.5&99.9&9.2&23.8&81.3&11.0&27.0&82.9\\
&$t_{16}(\check{\gamma})$&16.0&82.1&100.0&9.5&24.9&81.4&11.5&26.6&83.7\\
&$\tau(\check{\beta})$&16.7&67.7&94.4&9.1&13.8&50.9&10.3&19.2&60.0\\\hline
$\beta=0.012$&Cauchy RT&31.6&87.0&98.9&10.7&39.8&95.5&17.9&63.4&99.5\\
&Cauchy VC&47.3&94.5&98.5&12.1&58.3&95.4&14.9&53.5&98.7\\
&$t_8(\check{\gamma})$&33.0&96.7&100.0&13.9&50.8&99.1&15.4&51.7&99.4\\
&$t_{12}(\check{\gamma})$&35.1&98.2&100.0&13.3&52.2&99.5&15.1&55.5&99.8\\
&$t_{16}(\check{\gamma})$&34.8&98.3&100.0&13.2&51.2&99.1&15.4&53.5&99.8\\
&$\tau(\check{\beta})$&32.3&86.6&98.2&11.1&29.0&87.6&13.1&37.8&95.0\\\hline
$\beta=0.016$&Cauchy RT&52.7&95.2&99.7&14.2&65.3&99.3&26.2&90.6&100.0\\
&Cauchy VC&65.1&97.2&99.0&17.3&81.7&97.4&19.3&78.2&99.6\\
&$t_8(\check{\gamma})$&52.5&99.7&100.0&18.0&82.1&99.9&21.8&80.0&100.0\\
&$t_{12}(\check{\gamma})$&59.7&99.6&100.0&18.2&83.6&100.0&20.9&82.5&100.0\\
&$t_{16}(\check{\gamma})$&61.3&99.9&100.0&19.4&81.8&100.0&21.6&81.0&100.0\\
&$\tau(\check{\beta})$&51.6&94.3&98.9&13.5&54.9&96.6&17.4&59.5&99.5\\\hline
$\beta=0.02$&Cauchy RT&66.4&97.7&100.0&20.9&86.1&99.9&36.2&98.2&100.0\\
&Cauchy VC&76.8&98.0&99.2&21.5&91.1&98.7&25.0&92.1&99.7\\
&$t_8(\check{\gamma})$&74.3&100.0&100.0&23.7&95.4&99.9&27.8&94.7&100.0\\
&$t_{12}(\check{\gamma})$&80.4&99.9&100.0&23.6&96.3&100.0&28.4&95.8&100.0\\
&$t_{16}(\check{\gamma})$&81.1&100.0&100.0&25.7&96.8&100.0&29.1&96.3&100.0\\
&$\tau(\check{\beta})$&67.3&97.0&99.1&17.3&77.2&98.7&22.7&81.8&99.8\\
 \bottomrule\smallskip
\end{tabularx}
\end{center}
\scriptsize{
Rejection rates of the Cauchy RT, Cauchy VC, $t_q(\check{\gamma})$ ($q \in \{8,12,16\}$), and $\tau(\check{\beta})$ tests under $\beta \in \{0.004, 0.008, 0.012, 0.016, 0.02\}$  (two-sided tests). Data are from the continuous-time model \eqref{MC1} with $\bar{\kappa} \in \{0,5,20\}$, sampled at \(\Delta = 1/252\) (daily); \(T \in \{20,50,100\}\) years (\(252T\) observations). The volatility exhibits regime switching.}
\end{table}

\newpage

\begin{table}[h!]
\begin{center}
\footnotesize
\caption{Power for Continuous-Time Models with a Geometric Brownian Motion\label{tab5}}

\begin{tabularx}{0.84\textwidth}{ccCCCCCCCCC} \toprule
&&\multicolumn{3}{c}{$\bar{\kappa}=0$}&\multicolumn{3}{c}{$\bar{\kappa}=5$}&\multicolumn{3}{c}{$\bar{\kappa}=20$}\\
\cmidrule(r){3-5}\cmidrule(r){6-8}\cmidrule(r){9-11}
T&&20&50&100&20&50&100&20&50&100\\\hline
$\beta=0.004$&Cauchy RT&7.8&19.2&59.9&6.0&9.6&24.6&8.5&14.0&41.8\\
&Cauchy VC&42.4&59.7&80.4&6.7&13.0&29.5&7.2&10.8&19.6\\
&$t_8(\check{\gamma})$&11.0&27.1&67.7&4.4&7.5&21.4&5.7&8.4&18.8\\
&$t_{12}(\check{\gamma})$&10.3&28.5&72.0&5.4&8.8&21.7&4.9&8.1&19.2\\
&$t_{16}(\check{\gamma})$&10.2&29.2&73.2&5.5&9.3&21.2&5.8&9.2&19.9\\
&$\tau(\check{\beta})$&8.8&25.5&61.8&5.6&7.5&13.3&6.2&7.0&14.7\\\hline
$\beta=0.008$&Cauchy RT&13.7&59.1&92.4&8.2&24.7&73.4&13.2&40.9&92.8\\
&Cauchy VC&55.3&78.9&91.2&9.4&28.0&53.6&8.8&20.5&51.2\\
&$t_8(\check{\gamma})$&21.6&65.5&94.6&7.0&18.8&68.3&8.4&16.3&57.5\\
&$t_{12}(\check{\gamma})$&23.9&69.7&95.5&7.4&21.2&70.9&7.0&18.6&63.2\\
&$t_{16}(\check{\gamma})$&22.4&71.3&96.5&7.9&21.1&72.2&8.2&19.3&64.7\\
&$\tau(\check{\beta})$&21.2&59.2&88.0&6.7&12.0&35.0&8.0&14.5&40.5\\\hline
$\beta=0.012$&Cauchy RT&29.0&80.8&97.9&11.2&48.3&95.0&18.7&72.2&99.2\\
&Cauchy VC&65.6&88.3&92.2&13.8&42.9&65.6&12.1&36.1&69.9\\
&$t_8(\check{\gamma})$&38.9&85.5&97.3&10.7&43.0&89.6&10.4&35.2&81.2\\
&$t_{12}(\check{\gamma})$&40.2&88.4&98.5&11.0&43.6&93.0&10.5&39.8&88.0\\
&$t_{16}(\check{\gamma})$&42.0&89.0&99.1&10.0&44.4&92.8&11.1&39.6&91.5\\
&$\tau(\check{\beta})$&38.9&77.0&92.1&8.3&22.3&64.9&10.8&26.3&69.7\\\hline
$\beta=0.016$&Cauchy RT&47.9&91.7&99.1&16.0&70.5&98.5&27.8&91.5&99.7\\
&Cauchy VC&73.4&91.6&93.4&17.6&56.2&70.4&15.0&52.5&77.3\\
&$t_8(\check{\gamma})$&53.4&91.8&98.4&14.6&68.5&94.4&13.9&57.8&83.9\\
&$t_{12}(\check{\gamma})$&57.4&94.9&99.3&14.9&71.3&97.1&15.1&63.9&93.9\\
&$t_{16}(\check{\gamma})$&58.9&95.2&99.6&14.4&72.6&98.2&15.9&62.9&96.6\\
&$\tau(\check{\beta})$&51.2&86.3&93.5&10.3&37.4&74.8&13.1&42.1&82.3\\\hline
$\beta=0.02$&Cauchy RT&60.1&96.6&99.5&23.7&87.6&98.9&40.3&97.9&99.9\\
&Cauchy VC&79.7&92.6&93.9&22.3&64.9&72.9&19.4&63.6&80.7\\
&$t_8(\check{\gamma})$&67.5&94.8&99.0&20.2&81.4&96.0&18.3&71.0&85.5\\
&$t_{12}(\check{\gamma})$&69.9&97.2&99.7&19.7&85.7&98.6&20.0&80.7&95.0\\
&$t_{16}(\check{\gamma})$&72.9&97.2&99.9&20.2&87.5&99.1&21.2&79.9&98.2\\
&$\tau(\check{\beta})$&59.9&90.7&94.4&12.8&54.6&80.2&16.8&57.5&89.5\\
 \bottomrule\smallskip
\end{tabularx}
\end{center}
\scriptsize{
Rejection rates of the Cauchy RT, Cauchy VC, $t_q(\check{\gamma})$ ($q \in \{8,12,16\}$), and $\tau(\check{\beta})$ tests under $\beta \in \{0.004, 0.008, 0.012, 0.016, 0.02\}$  (two-sided tests). Data are from the continuous-time model \eqref{MC1} with $\bar{\kappa} \in \{0,5,20\}$, sampled at \(\Delta = 1/252\) (daily); \(T \in \{20,50,100\}\) years (\(252T\) observations). The volatility follows a geometric Brownian motion.}
\end{table}

\newpage

\begin{table}[h!]
\begin{center}
\scriptsize
\caption{Size and Power for Discrete-Time Models with Constant Volatility and MA(1) Innovations\label{tab6}}

\begin{tabularx}{0.84\textwidth}{ccCCCCCCCCC} \toprule
&&\multicolumn{3}{c}{$\bar{\kappa}=0$}&\multicolumn{3}{c}{$\bar{\kappa}=50$}&\multicolumn{3}{c}{$\bar{\kappa}=100$}\\
\cmidrule(r){3-5}\cmidrule(r){6-8}\cmidrule(r){9-11}
T&&20&50&100&20&50&100&20&50&100\\\hline
$\beta=0$&$\tau(\check{\beta}_o)$&5.0&4.9&4.8&5.2&5.3&4.8&4.5&4.7&5.2\\
&$\tau(\check{\beta}_e)$&4.8&4.7&5.2&5.0&4.8&4.8&5.1&5.0&4.9\\
&IVX&14.2&13.5&12.3&10.2&10.5&10.3&9.5&10.3&10.5\\\hline
$\beta=0.5$&$\tau(\check{\beta}_o)$&8.4&7.9&8.4&17.0&17.5&16.6&24.9&24.9&25.8\\
&$\tau(\check{\beta}_e)$&8.2&8.0&8.4&17.4&16.0&17.1&25.4&25.7&26.0\\
&IVX&100.0&100.0&100.0&100.0&100.0&100.0&100.0&100.0&100.0\\\hline
$\beta=1$&$\tau(\check{\beta}_o)$&13.4&13.4&13.9&39.3&41.1&39.7&60.5&61.1&62.0\\
&$\tau(\check{\beta}_e)$&13.7&13.5&13.7&39.6&39.1&39.6&62.0&62.8&62.5\\
&IVX&100.0&100.0&100.0&100.0&100.0&100.0&100.0&100.0&100.0\\\hline
$\beta=1.5$&$\tau(\check{\beta}_o)$&20.2&20.1&20.8&63.3&65.2&65.0&85.9&87.6&87.7\\
&$\tau(\check{\beta}_e)$&20.7&20.4&20.6&63.3&64.0&64.3&86.5&87.8&87.8\\
&IVX&100.0&100.0&100.0&100.0&100.0&100.0&100.0&100.0&100.0\\\hline
$\beta=2$&$\tau(\check{\beta}_o)$&27.7&27.3&28.3&80.3&82.1&82.9&95.4&96.6&97.2\\
&$\tau(\check{\beta}_e)$&28.0&27.3&28.4&80.1&81.4&82.6&95.9&97.0&97.1\\
&IVX&100.0&100.0&100.0&100.0&100.0&100.0&100.0&100.0&100.0\\\hline
$\beta=2.5$&$\tau(\check{\beta}_o)$&34.3&33.7&34.7&89.5&91.7&92.1&98.6&99.2&99.3\\
&$\tau(\check{\beta}_e)$&34.8&33.8&35.4&89.1&91.2&92.0&98.7&99.1&99.3\\
&IVX&100.0&100.0&100.0&100.0&100.0&100.0&100.0&100.0&100.0\\
 \bottomrule\smallskip
\end{tabularx}
\end{center}
\scriptsize{
Rejection rates of $\tau(\check{\beta}_o)$, $\tau(\check{\beta}_e)$, and IVX tests under $\beta \in \{0, 0.5, 1, 1.5, 2, 2.5\}$  (one-sided tests). Data are from the discrete-time model \eqref{MC2}--\eqref{MC3} with $\bar{\kappa} \in \{0,50,100\}$ and MA(1) innovations, sampled monthly; $T \in \{20,50,100\}$ years ($12T$ observations). The volatility is constant.
}
\end{table}

\newpage

\begin{table}[h!]
\scriptsize
\caption{Size and Power for Discrete-Time Models with Structural Breaks in Volatility and MA(1) Innovations\label{tab8}}

\begin{center}
\begin{tabularx}{0.84\textwidth}{ccCCCCCCCCC} \toprule
&&\multicolumn{3}{c}{$\bar{\kappa}=0$}&\multicolumn{3}{c}{$\bar{\kappa}=50$}&\multicolumn{3}{c}{$\bar{\kappa}=100$}\\
\cmidrule(r){3-5}\cmidrule(r){6-8}\cmidrule(r){9-11}
T&&20&50&100&20&50&100&20&50&100\\\hline
$\beta=0$&$\tau(\check{\beta}_o)$&4.5&4.9&5.1&5.3&5.2&4.8&5.1&4.7&4.7\\
&$\tau(\check{\beta}_e)$&4.8&4.8&4.9&5.0&4.5&4.8&5.0&4.9&4.9\\
&IVX&30.7&31.2&30.1&32.8&34.2&34.3&33.9&34.5&36.0\\\hline
$\beta=0.5$&$\tau(\check{\beta}_o)$&7.7&8.1&8.7&13.7&13.7&13.2&18.5&18.7&19.3\\
&$\tau(\check{\beta}_e)$&8.0&8.2&8.4&13.6&13.0&12.8&19.3&19.7&19.7\\
&IVX&100.0&100.0&100.0&100.0&100.0&100.0&100.0&100.0&100.0\\\hline
$\beta=1$&$\tau(\check{\beta}_o)$&13.5&13.5&14.3&29.2&29.3&28.5&44.3&45.0&46.0\\
&$\tau(\check{\beta}_e)$&14.0&13.8&13.9&28.2&28.1&27.3&46.1&46.6&46.1\\
&IVX&100.0&100.0&100.0&100.0&100.0&100.0&100.0&100.0&100.0\\\hline
$\beta=1.5$&$\tau(\check{\beta}_o)$&20.4&20.4&21.5&46.8&48.1&47.6&68.9&70.8&72.2\\
&$\tau(\check{\beta}_e)$&21.3&20.8&21.0&46.9&46.3&47.5&71.2&71.5&72.5\\
&IVX&100.0&100.0&100.0&100.0&100.0&100.0&100.0&100.0&100.0\\\hline
$\beta=2$&$\tau(\check{\beta}_o)$&27.6&27.7&28.9&62.7&64.9&64.7&84.4&87.3&88.0\\
&$\tau(\check{\beta}_e)$&28.3&28.0&28.0&63.2&63.2&64.5&85.7&87.3&87.9\\
&IVX&100.0&100.0&100.0&100.0&100.0&100.0&100.0&100.0&100.0\\\hline
$\beta=2.5$&$\tau(\check{\beta}_o)$&34.4&34.3&35.5&74.0&76.9&77.2&92.1&94.6&95.4\\
&$\tau(\check{\beta}_e)$&35.1&34.6&34.3&74.7&75.9&77.3&93.0&94.6&95.3\\
&IVX&100.0&100.0&100.0&100.0&100.0&100.0&100.0&100.0&100.0\\
 \bottomrule\smallskip
\end{tabularx}
\end{center}
\scriptsize{
Rejection rates of $\tau(\check{\beta}_o)$, $\tau(\check{\beta}_e)$, and IVX tests under $\beta \in \{0, 0.5, 1, 1.5, 2, 2.5\}$  (one-sided tests). Data are from the discrete-time model \eqref{MC2}--\eqref{MC3} with $\bar{\kappa} \in \{0,50,100\}$ and MA(1) innovations, sampled monthly; $T \in \{20,50,100\}$ years ($12T$ observations). The volatility exhibits structural breaks.
}
\end{table}

\newpage

\begin{table}[h!]
\begin{center}
\scriptsize
\caption{Size and Power for Discrete-Time Models with Regime Switching in Volatility and MA(1) Innovations\label{tab10}}

\begin{tabularx}{0.84\textwidth}{ccCCCCCCCCC} \toprule
&&\multicolumn{3}{c}{$\bar{\kappa}=0$}&\multicolumn{3}{c}{$\bar{\kappa}=50$}&\multicolumn{3}{c}{$\bar{\kappa}=100$}\\
\cmidrule(r){3-5}\cmidrule(r){6-8}\cmidrule(r){9-11}
T&&20&50&100&20&50&100&20&50&100\\\hline
$\beta=0$&$\tau(\check{\beta}_o)$&3.2&4.0&4.3&4.7&5.0&5.0&5.2&4.9&4.8\\
&$\tau(\check{\beta}_e)$&3.8&3.9&4.2&5.0&4.9&4.9&4.9&5.1&4.8\\
&IVX&13.5&13.2&12.4&9.7&11.1&11.0&10.3&11.2&12.4\\\hline
$\beta=0.5$&$\tau(\check{\beta}_o)$&5.8&6.8&7.4&16.1&17.1&16.6&23.4&24.1&24.4\\
&$\tau(\check{\beta}_e)$&6.4&6.7&6.9&16.8&16.8&16.6&23.9&24.9&23.6\\
&IVX&100.0&100.0&100.0&100.0&100.0&100.0&100.0&100.0&100.0\\\hline
$\beta=1$&$\tau(\check{\beta}_o)$&10.4&11.3&12.0&39.6&40.9&39.3&56.7&60.4&60.6\\
&$\tau(\check{\beta}_e)$&11.0&11.7&11.8&39.7&39.6&39.8&57.3&60.9&59.6\\
&IVX&100.0&100.0&100.0&100.0&100.0&100.0&100.0&100.0&100.0\\\hline
$\beta=1.5$&$\tau(\check{\beta}_o)$&17.3&18.2&19.0&64.3&67.1&63.9&81.5&86.0&87.0\\
&$\tau(\check{\beta}_e)$&17.9&18.6&18.7&63.8&65.3&64.4&82.2&86.8&86.2\\
&IVX&100.0&100.0&100.0&100.0&100.0&100.0&100.0&100.0&100.0\\\hline
$\beta=2$&$\tau(\check{\beta}_o)$&25.0&25.4&26.5&80.2&83.1&81.7&92.0&95.5&96.4\\
&$\tau(\check{\beta}_e)$&25.6&26.2&25.9&79.4&82.0&82.1&92.5&95.5&96.1\\
&IVX&100.0&100.0&100.0&100.0&100.0&100.0&100.0&100.0&100.0\\\hline
$\beta=2.5$&$\tau(\check{\beta}_o)$&32.6&32.5&33.9&88.9&91.4&90.7&96.5&98.6&99.1\\
&$\tau(\check{\beta}_e)$&33.4&33.6&32.7&88.2&90.9&90.6&96.8&98.6&98.9\\
&IVX&100.0&100.0&100.0&100.0&100.0&100.0&100.0&100.0&100.0\\
 \bottomrule\smallskip
\end{tabularx}
\end{center}
\scriptsize{
Rejection rates of $\tau(\check{\beta}_o)$, $\tau(\check{\beta}_e)$, and IVX tests under $\beta \in \{0, 0.5, 1, 1.5, 2, 2.5\}$  (one-sided tests). Data are from the discrete-time model \eqref{MC2}--\eqref{MC3} with $\bar{\kappa} \in \{0,50,100\}$ and MA(1) innovations, sampled monthly; $T \in \{20,50,100\}$ years ($12T$ observations). The volatility exhibits regime switching.
}
\end{table}

\newpage

\begin{table}[h!]
\begin{center}
\footnotesize
\caption{Empirical Results on Stock Return Predictability (One-Sided Tests)\label{tab12}}

\begin{tabularx}{0.84\textwidth}{lllllllll} \toprule
Series &Frequency&$\beta$&CRT&CVC&$t_{12}(\check{\gamma})$&$t_{16}(\check{\gamma})$&$\tau(\check{\beta})$&IVX\\\hline
\multicolumn{9}{l}{Panel A: D/P as predictor for the period of 1927-2011}\\\hline
CRSP&Monthly&0.005&1.06&0.85&0.42&0.44&0.43&$1.93^{*}$\\
&Quarterly&0.007&0.48&1.26&1.29&1.34&1.09&$2.45^{**}$\\
&Yearly&0.063&0.99&$2.15^{*}$&$2.04^{*}$&$2.35^{*}$&$1.82^{*}$&$2.37^{**}$\\
S\&P500&Monthly&0.003&0.74&1.14&0.89&0.82&0.72&$1.99^{*}$\\
&Quarterly&0.008&0.65&$1.78^{*}$&$1.75^{}$&$1.73^{}$&$1.31^{}$&$2.64^{**}$\\
&Yearly&0.042&0.82&\textcolor{black}{$2.24^{*}$}&$2.33^{*}$&$2.00^{*}$&\textcolor{black}{$1.95^{*}$}&$1.85^{*}$\\\hline
\multicolumn{9}{l}{Panel B: D/P as predictor for the period of 1927-2011 with jumps removed}\\\hline
CRSP&Monthly&0.001&0.22&$1.50^{}$&1.15&1.06&0.90&$2.08^{*}$\\
&Quarterly&0.015&0.93&$2.02^{*}$&\textcolor{black}{$2.45^{*}$}&\textcolor{black}{$2.40^{*}$}&$1.46^{}$&$2.67^{**}$\\
&Yearly&0.005&0.05&0.56&0.56&0.53&0.45&1.34\\
S\&P500&Monthly&0.002&0.48&$1.66^{*}$&1.35&1.14&1.02&$2.08^{*}$\\
&Quarterly&0.017&1.21&$1.49^{}$&1.70&1.61&1.20&$2.38^{**}$\\
&Yearly&0.008&0.11&\textcolor{black}{1.51}&$2.45^{*}$&$2.56^{*}$&$1.80^{*}$&$1.67^{*}$\\\hline
\multicolumn{9}{l}{Panel C: E/P as predictor for the period of 1950-2011}\\\hline
S\&P500&Monthly&0.000&-0.05&0.32&0.26&0.22&0.30&0.96\\
&Quarterly&0.007&0.39&0.39&0.38&0.35&0.35&1.06\\
&Yearly&0.059&0.83&0.76&1.04&0.23&0.84&1.11\\\hline
\multicolumn{9}{l}{Panel D: E/P as predictor for the period of 1950-2011 with jumps removed}\\\hline
S\&P500&Monthly&0.001&0.16&0.01&0.14&0.18&0.11&0.76\\
&Quarterly&0.000&0.02&0.71&0.78&0.61&0.48&0.91\\
&Yearly&-0.038&-0.38&0.69&-0.05&0.55&0.44&1.02\\
 \bottomrule\smallskip
\end{tabularx}
\end{center}
\scriptsize{Results for one-sided tests of return predictability for the NYSE/AMEX value-weighted index (CRSP) and the S\&P~500 using the Cauchy RT (CRT), Cauchy VC (CVC), \(t_q\) (\(q = 12,16\)), \(\tau(\check{\beta})\), and IVX tests across different regression frequencies. Panels~A--B use the dividend--price ratio (D/P), while Panels~C--D use the earnings--price ratio (E/P) as predictors. Statistical significance at the 5\% and 1\% levels is denoted by ``$^{*}$'' and ``$^{**}$'', respectively.}
\end{table}

\bibliographystyle{agsm} 
\bibliography{Predictive}

\pagebreak

\newpage

\clearpage

\appendix

\begin{center}
\vspace*{2in}{\Huge Online Appendix for}

\bigskip

{\Huge \bigskip}

{\Large Robust Cauchy-Based Methods for Predictive Regressions}

\begin{abstract}
This supplement includes four appendices, Appendices~A--D. Appendix~A contains additional theoretical analysis, Appendices~B and C provide additional simulation and empirical results, respectively, and Appendix~D collects the proofs of the results stated in the main paper.
\end{abstract}
\end{center}

\setcounter{equation}{0}
\setcounter{figure}{0}
\renewcommand{\theequation}{A.\arabic{equation}}%

\renewcommand{\theass}{A.\arabic{ass}}


\section{Additional Theoretical Analysis}

\subsection{Wald Test for Joint Predictability Test}

Consider a predictive regression with $K$ predictors $x_t = [x_{1,t}, \ldots, x_{K,t}]'$:
\begin{align}
y_t &= x_{t-1}' B + u_t \notag \\
    &= \beta_{1,K} x_{1,t-1} + \cdots + \beta_{K,K} x_{K,t-1} + u_t, 
    \qquad B = [\beta_{1,K}, \ldots, \beta_{K,K}]'. \label{PredRegM}
\end{align}
The objective is to test the joint predictability of the covariates, that is,
\[
H_0: \beta_{1,K}=\cdots=\beta_{K,K} =0.
\]
Define
$z_t = [\operatorname{sign}(x_{1,t}), \ldots, \operatorname{sign}(x_{K,t})]'$. Under $H_0$ we have $y_t=u_t$, so the normalized score is the vector martingale
\[
S_T = \frac{1}{\sqrt{T}}\sum_{t=1}^T z_{t-1} y_t = \frac{1}{\sqrt{T}}\sum_{t=1}^T z_{t-1} u_t,
\]
whose predictable conditional covariance is $\frac1T\sum_{t=1}^T v_t^2 z_{t-1}z_{t-1}'$. The diagonal entries converge to $\omega^2=\int_0^1\sigma^2(r)\,dr$ (since $\operatorname{sign}(\cdot)^2=1$), but for persistent regressors the off-diagonal entries $\frac1T\sum_{t=1}^T v_t^2\operatorname{sign}(x_{j,t-1})\operatorname{sign}(x_{k,t-1})$ have random limits. Consequently $S_T$ is in general \emph{mixed} normal, $S_T\to_d\mathbb{MN}(0,\Sigma)$ with a (possibly random) limiting covariance $\Sigma$, rather than Gaussian. Under the additional conditions in Assumption~\ref{ass-a1} below -- which ensure that the self-normalizer is consistent for $\Sigma$ and that the limiting Gaussian is independent of $\Sigma$ -- the studentized score is asymptotically pivotal,
\[
\left(\hat{\omega}^2 \sum_{t=1}^T z_{t-1} z_{t-1}' \right)^{-1/2}
\left( \sum_{t=1}^T z_{t-1} y_t \right)
\to_d \mathbb{N}(0, I_K),
\]
where $\hat{\omega}^2\to_p \omega^2$ and the limit of $\frac1T\sum_{t=1}^T z_{t-1} z_{t-1}'$ is invertible. This convergence motivates the Wald-type statistic $W$ defined as
\[
W = 
\left( \sum_{t=1}^T z_{t-1} y_t \right)' 
\left( \hat{\omega}^2 \sum_{t=1}^T z_{t-1} z_{t-1}' \right)^{-1}
\left( \sum_{t=1}^T z_{t-1} y_t \right),
\]
where 
\[
\hat{\omega}^2 = \frac{1}{T}\sum_{t=1}^T \hat{u}_t^2,
\qquad
\hat{u}_t = y_t - x_{t-1}'\hat{B}
\]
with the OLS estimator $\hat{B}$. As in the main paper, $\hat{\omega}^2$ estimates $\omega^2=\int_0^1\sigma^2(r)\,dr$ using OLS residuals.

For our purpose, we assume 
\begin{ass}\phantomsection\label{ass-a1}
\begin{enumerate}
\item[(i)] $\frac{1}{T}\sum_{t=1}^T z_{t-1} z_{t-1}'\to_d M_z$, where $M_z$ is invertible.
\item[(ii)] $\left\Vert\sum_{t=1}^T x_{t-1}u_t \right\Vert_2= o_p\!\left(\sqrt{T\lambda_T}\right)$, where $\lambda_T$ is the minimum eigenvalue of $\sum_{t=1}^T x_{t-1}x_{t-1}'$ and the norm $\Vert \cdot \Vert_2$ is the spectral norm.
\item[(iii)] let $\varsigma_T$ be the minimum singular value of $\sum_{t=1}^T z_{t-1} x_{t-1}'$ such that $\varsigma_T/\sqrt{T}\to_p\infty$.
\item[(iv)] $\left\Vert \hat{\omega}^2\,\dfrac1T\sum_{t=1}^T z_{t-1}z_{t-1}' - \dfrac1T\sum_{t=1}^T v_t^2\, z_{t-1}z_{t-1}'\right\Vert_2 \to_p 0$, where $\Vert\cdot\Vert_2$ is the spectral norm.
\item[(v)] one of the following holds: \emph{(v-a)} (stationary regressors) the convergence in (i) holds in probability with $M_z$ deterministic; or \emph{(v-b)} (persistent regressors) the partial-sum processes converge jointly,
\[
\left(\frac1{\sqrt T}\sum_{t\le[T\cdot]} z_{t-1}u_t,\ \frac1T\sum_{t\le[T\cdot]} v_t^2 z_{t-1}z_{t-1}'\right)\Rightarrow (S,\,V),
\]
where $S$ is a continuous martingale with quadratic variation $\langle S\rangle = V$, admitting the representation $S(1)=\int_0^1\sigma(r)\,\tilde z(r)\,dB(r)$ in which $B$ is a scalar standard Brownian motion, $\sigma(\cdot)$ is the spot-volatility process of the main paper, and $\tilde z(\cdot)$ is the weak limit of the sign process $z_{[T\cdot]}$; the integrand $(\sigma(\cdot),\tilde z(\cdot))$ is independent of $B$, so that the limiting quadratic-variation process is $V(r)=\int_0^r\sigma^2(s)\,\tilde z(s)\tilde z(s)'\,ds$ and the mixing covariance of the motivating mixed-normal limit is its terminal value $\Sigma=V(1)=\int_0^1\sigma^2(r)\,\tilde z(r)\tilde z(r)'\,dr$.
\end{enumerate}
\end{ass}

Assumption \ref{ass-a1} is essentially a multivariate extension of the conditions used for Proposition 3.4 in the main paper as we will discuss subsequently. Firstly, Assumption \ref{ass-a1}(i) generalizes the univariate case, in which we always have $\sum_{t=1}^T z_{t-1}^2 = T$ for all $T$ with probability one. Secondly, Assumption \ref{ass-a1}(ii) extends Assumption 3.1 in the main paper. Lastly, Assumption \ref{ass-a1}(iii) is a multivariate version of the test consistency assumption $\sum_{t=1}^T |x_{t-1}|/\sqrt{T}\to_p\infty$ used in Proposition 3.4(b) in the main paper. Conditions (iv) and (v) are the multivariate counterparts of the asymptotic volatility-exogeneity built into Assumption~2.3 of the main paper. Condition (iv) requires the self-normalizing matrix $\hat{\omega}^2\frac1T\sum_{t=1}^T z_{t-1}z_{t-1}'$ to be consistent, in spectral norm, for the predictable conditional covariance $\frac1T\sum_{t=1}^T v_t^2 z_{t-1}z_{t-1}'$. Condition (v) ensures that the limiting Gaussian is independent of the (possibly random) limiting covariance, so that the mixed-normal limit is a genuine scale mixture and the studentized statistic is pivotal; it separates two regimes. For stationary regressors (v-a), the sign Gram matrix $\frac1T\sum_{t=1}^T z_{t-1}z_{t-1}'$ has a deterministic limit $M_z$, so the joint dependence across signs is asymptotically non-random; the Wald-specific requirement then reduces to the covariance-consistency condition (iv). We stress that, under persistent volatility, condition (iv) itself embeds an asymptotic volatility-exogeneity requirement, since it forces the volatility-weighted Gram $\frac1T\sum_{t=1}^T v_t^2 z_{t-1}z_{t-1}'$ to share the limit $\omega^2 M_z$; the stationary case is therefore not assumption-free, but the exogeneity content is confined to (iv). For persistent regressors (v-b), this factorization no longer holds, and condition (v) imposes the stronger, process-level counterpart of the independence $\sigma\perp W$ used in the univariate analysis of the main paper. Condition (v-b) holds, for example, when the regressors are driven by innovations independent of the martingale difference $(\varepsilon_t)$ that drives the returns, so that the limiting sign structure is independent of $B$; it fails in the canonical predictive-regression design with persistent and endogenous regressors, where the regressor and return innovations are correlated -- the case reflected in the size results at $\bar{\kappa}=0$ in Appendix~B.3.

Note that Assumption~\ref{ass-a1} is violated if $x_t$ is supported only on the positive (or negative) real line. In this case, one may subtract $m_t$, where $m_t$ is an $\mathcal{F}_{t-1}$-adapted process taking values in the support of $x_t$, and apply the test to the transformed regressor $x_t - m_t$ (see Footnote~9 of the main paper). Moreover, when an intercept is included, the Wald test can be combined with the first-differencing approach described in Section~4.2 of the main paper.

This formulation helps clarify the relationship between our Wald statistic, the Bonferroni procedure in the main paper, and IVX-type inference. The Bonferroni procedure inherits the robustness of the univariate Cauchy statistics because it does not require modeling the joint dependence of the signs across predictors. The Wald statistic deliberately uses this joint dependence, and hence \textcolor{black}{can be more efficient}, but only when the covariance and independence requirements in Assumption~\ref{ass-a1} are credible\textcolor{black}{; in its present form, it also addresses only the joint null (see below)}. IVX methods, by contrast, are built for multivariate predictive regressions with persistent endogenous regressors and general linear restrictions (see \citeauthor{KMS2015}, \citeyear{KMS2015}; \citeauthor{KMS2023}, \citeyear{KMS2023}), but they do not target the same robustness to heavy tails and persistent or nonstationary volatility.

The proposed Wald test is not straightforward to extend to general linear hypotheses, such as $R\beta = r$ for some $q \times K$ restriction matrix $R$ and $q \times 1$ vector $r$. This is because the asymptotic distribution of the restricted Cauchy-type IV estimator is complex and, in general, depends on the degrees of persistence and heavy-tailedness of the covariates $x_t$. Accordingly, our Wald statistic is limited to joint predictability tests for the full set of covariates in the predictive regression, whereas the IVX method of \citet{KMS2015} can be applied to test general linear restrictions under their maintained assumptions.

Nonetheless, testing joint significance for a subset of parameters may be feasible by combining our Wald test with a partialling-out approach based on the Frisch--Waugh--Lovell theorem, under additional regularity conditions. We leave such generalizations of our Wald-type approach for future research.

\begin{prop}\label{prop-a1}
Let Assumptions~2.1, 2.2, and 2.3 in the main paper hold.  \smallskip

(a) Let Assumption \ref{ass-a1}(i), (ii), (iv), and (v) hold. Then, under $H_0:B=\mathbf{0}_{K\times1}$,  $W\to_d\chi_K^2$.\smallskip

(b) Let Assumption \ref{ass-a1}(i)-(iii) hold.  Then, under $H_A:B\neq\mathbf{0}_{K\times1}$, $W\to_p\infty$.
\end{prop}

\begin{proof}
Let $\omega_T^2 = T^{-1} \sum_{t=1}^T u_t^2$.  
Then $\omega_T^2 \to_p \omega^2$ as in the proof of Proposition 3.4 of the main paper. Furthermore,  we have
\[
\hat{\omega}^2 
= \omega_T^2 - \left(\sum_{t=1}^T x_{t-1} u_t\right)'\left(T\sum_{t=1}^T x_{t-1}x_{t-1}'\right)^{-1}\left(\sum_{t=1}^T x_{t-1} u_t\right)
= \omega_T^2 + o_p(1)
\]
by Assumption \ref{ass-a1}(ii). Under $H_0$, $\sum_{t=1}^T z_{t-1}y_t=\sum_{t=1}^T z_{t-1}u_t$, and $S_T=T^{-1/2}\sum_{t=1}^T z_{t-1}u_t$ is a vector martingale with predictable quadratic variation $\frac1T\sum_{t=1}^T v_t^2 z_{t-1}z_{t-1}'$. By the martingale central limit theorem \citep[Theorem~3.2]{HallHeyde1980}, using the conditional Lindeberg condition implied by Assumption~2.1 together with condition (v), $S_T\to_d\mathbb{MN}(0,\Sigma)=\Sigma^{1/2}N$, where $\Sigma$ is the limit of the predictable quadratic variation and $N\sim\mathbb{N}(0,I_K)$ is independent of $\Sigma$; under (v-a) the limit $\Sigma=\omega^2 M_z$ is attained in probability, while under (v-b) it follows from the joint stochastic-integral convergence of \citet{hansen-1992}, \citet{kurtz1991weak} and \citet{IP}. By condition (iv), the self-normalizer satisfies $\hat{\omega}^2\frac1T\sum_{t=1}^T z_{t-1}z_{t-1}'\to_d\Sigma$ jointly with $S_T$, with the same limit $\Sigma$. Hence
\[
\Big(\hat{\omega}^2 \sum_{t=1}^T z_{t-1}z_{t-1}'\Big)^{-1/2}\sum_{t=1}^T z_{t-1}y_t
=\Big(\hat{\omega}^2 \tfrac1T\sum_{t=1}^T z_{t-1}z_{t-1}'\Big)^{-1/2} S_T
\to_d \Sigma^{-1/2}\Sigma^{1/2}N = N,
\]
and therefore $W\to_d N'N=\chi_K^2$, which establishes part (a).

As for Part (b), we note that under $H_A$
\[
W =_d \left( \frac{1}{\sqrt{T}}\sum_{t=1}^T z_{t-1} x_{t-1}'B \right)' 
\left( \textcolor{black}{\omega^2} M_z \right)^{-1}
\left( \frac{1}{\sqrt{T}}\sum_{t=1}^T z_{t-1} x_{t-1}'B \right)(1+o_p(1)).
\]
Therefore, $W\to_p\infty$ as long as 
\[
\frac{1}{\sqrt{T}}\sum_{t=1}^T z_{t-1} x_{t-1}'B\to_p\infty
\]
for any fixed $B\neq\mathbf{0}_{K\times1}$ in $\mathbb{R}^K$, which is implied by Assumption \ref{ass-a1}(iii). This completes the proof.
\end{proof}

\begin{remark}\label{remark-wald-scope}
Condition (v) clarifies how the two multivariate procedures divide labor. For stationary regressors (case (v-a)) the limiting sign Gram matrix is deterministic, so the additional independence requirement in (v-b) is unnecessary; subject to the covariance consistency requirement in condition (iv), the Wald test can therefore be used with endogenous stationary regressors. For persistent (near unit root) regressors (case (v-b)), the Wald test trades part of the univariate robustness for efficiency: because its cross terms reintroduce the regressors into the limiting variance -- unlike each univariate Cauchy statistic, whose quadratic variation is free of the regressors since $\operatorname{sign}(\cdot)^2=1$ -- its validity then requires the limiting sign structure to be asymptotically independent of the innovations, that is, the persistent regressors to be asymptotically exogenous. The component-wise Bonferroni procedure of Section~4.1 of the main paper imposes no such requirement and therefore remains valid for endogenous, persistent regressors; it is the robust default for joint predictability testing, while the Wald test \textcolor{black}{is the more efficient option when the regressors are stationary with consistent covariance estimation or asymptotically exogenous; extending it beyond the joint null, for example to conditional inference on individual coefficients, is left for future research}.
\end{remark}


\subsection{Effects of First Differencing on Testing Power}

Corollary 4.2(b) of the main paper establishes the consistency of the test based on the first-differenced model. However, first differencing generally entails a loss of power relative to the test that imposes a zero intercept. In this subsection, we provide a theoretical comparison of the asymptotic power of the two tests.

Specifically, consider the linear predictive regression without intercept,
\[
y_t = \beta x_{t-1} + u_t,
\]
where $(x_t)$ follows the local-to-unity process of \cite{Phillips_Magdalinos_2007book},
\[
x_t = \left(1+\frac{c}{T^\delta}\right) x_{t-1} + \varepsilon_t^x, 
\quad c<0, \quad \delta \in [0,1],
\]
and $(\varepsilon_t^x)$ satisfies Assumption LP therein. We assume that Assumptions 2.1--2.3 in the main paper hold. 

We compare the rates of divergence of the test statistics $\tau(\check{\beta})$ and $\tau(\check{\beta}_e)$, defined as
\[
\tau(\check{\beta}) 
= \frac{\sum_{t=1}^T \text{sign}(x_{t-1}) y_t}{\hat{\omega}\sqrt{T}}
\quad \text{and}\quad
\tau(\check{\beta}_e) 
= \frac{\sum_{t=2}^{T/2} \text{sign}(x_{2t-2})(y_{2t} - y_{2t-1})}{\hat{\omega}_d \sqrt{T/2}},
\]
where $\hat{\omega}^2 = T^{-1} \sum_{t=1}^T \hat{u}_t^2$ and $\hat{\omega}_d^2 = T^{-1} \sum_{t=1}^T (\hat{u}_t^{(d)})^2$, with $(\hat{u}_t)$ denoting the OLS residuals from the original model and $(\hat{u}_t^{(d)})$ those from the demeaned model.

It follows from Proposition 3.4(b) and Corollary 4.2(b) of the main paper that, under $\beta \neq 0$,
\[
\tau(\check{\beta})
= \beta \frac{\sum_{t=1}^T |x_{t-1}|}{\omega \sqrt{T}} (1 + o_p(1)) + O_p(1),
\]
and
\begin{align*}
\tau(\check{\beta}_e)
= \beta \frac{\sum_{t=1}^{T/2} \text{sign}(x_{2t-2})(x_{2t-1} - x_{2t-2})}
{\omega \sqrt{T/2}} (1 + o_p(1)) + O_p(1),
\end{align*}
provided that Assumption 3.1 holds for both $x_{t-1}$ and $x_{t-1} - \bar{x}_T$.

The rate of divergence for $\tau(\check{\beta})$ follows directly from Lemma 3.2 of \cite{Phillips_Magdalinos_2007book}:
\[
\frac{1}{T^{1/2 + \delta/2}} \tau(\check{\beta})
= \frac{\beta}{\omega} M_1 (1 + o_p(1)),
\]
where $M_k =  E|X_k(1)|$ for $k=1,2$, and
\[
X_k(t) = \int_0^t \exp\bigl(kc (t-r)\bigr) dB_k(r),
\]
with $(B_k)$ denoting Brownian motions defined in Assumption LP of \cite{Phillips_Magdalinos_2007book}.

For $\tau(\check{\beta}_e)$, we decompose
\[
\tau(\check{\beta}_e) 
= \frac{\beta}{\omega/\sqrt{2}} \bigl[\tau_M(\check{\beta}_e) + \tau_R(\check{\beta}_e)\bigr] (1 + o_p(1)) + O_p(1),
\]
where
\begin{align*}
\tau_M(\check{\beta}_e) 
= \frac{c}{T^{1/2 + \delta}} \sum_{t=1}^{T/2} |x_{2t-2}|, \quad 
\tau_R(\check{\beta}_e) 
= \frac{1}{\sqrt{T}} \sum_{t=1}^{T/2} \text{sign}(x_{2t-2}) \varepsilon_{2t-1}^x,
\end{align*}
using the decomposition
\begin{align*}
\sum_{t=1}^{T/2} \text{sign}(x_{2t-2})(x_{2t-1} - x_{2t-2})
&= \frac{c}{T^\delta} \sum_{t=1}^{T/2} |x_{2t-2}| 
+ \sum_{t=1}^{T/2} \text{sign}(x_{2t-2}) \varepsilon_{2t-1}^x.
\end{align*}

Using the second-order representation
\[
x_t 
= \left(1+\frac{2c}{T^\delta} + \frac{c^2}{T^{2\delta}} \right)x_{t-2} 
+ \varepsilon_t^x + \left(1+\frac{c}{T^\delta}\right)\varepsilon_{t-1}^x,
\]
Lemma 3.2 of \cite{Phillips_Magdalinos_2007book} yields
\[
\frac{1}{T^{1/2 - \delta/2}} \tau_M(\check{\beta}_e)
= M_2 (1 + o_p(1)),
\]
so that $\tau_M(\check{\beta}_e)$ diverges at rate $T^{1/2 - \delta/2}$ whenever $\delta < 1$. For the remainder term, $\tau_R(\check{\beta}_e)$ diverges at rate $\sqrt{T}$ if 
\[
 E[\text{sign}(x_{2t-2}) \varepsilon_{2t-1}^x] \neq 0,
\]
and otherwise $\tau_R(\check{\beta}_e) = O_p(1)$.

We summarize the results as follows:
\begin{itemize}
    \item $\tau(\check{\beta})$ is always consistent and diverges at rate $T^{1/2 + \delta/2}$.

    \item $\tau(\check{\beta}_e)$ is consistent if either $\delta < 1$ or 
    $ E[\text{sign}(x_{2t-2}) \varepsilon_{2t-1}^x] \neq 0$. Its divergence rate is $T^{1/2 - \delta/2}$ when the expectation is zero, and $\sqrt{T}$ otherwise.

    \item Consequently, the divergence rate of $\tau(\check{\beta})$ is uniformly faster than that of $\tau(\check{\beta}_e)$, implying a loss of power induced by first differencing.
\end{itemize}

An important observation is that the test based on $\tau(\check{\beta}_e)$ is inconsistent when $\delta = 1$ and $ E[\text{sign}(x_{2t-2}) \varepsilon_{2t-1}^x] = 0$. A canonical example is the pure random walk $x_t$, i.e., $x_t = x_{t-1} + \varepsilon_t^x$, with $\varepsilon_t^x \sim \text{i.i.d. }(0,\sigma^2)$.

\newpage
\section{Additional Simulation Results}
\subsection{Discrete-Time Experiments: One-Sided Tests}

Following Section~5.2 of the main paper, we examine the finite-sample performance of the proposed tests in a discrete-time setting with an intercept, comparing them to the IVX test of \citet{KMS2015}. The data-generating process (DGP) is specified as
\begin{align}
y_t &= \beta x_{t-1} + \sigma_t \varepsilon_t, \label{MC2} \\
x_t &= \left(1 - \frac{\bar{\kappa}}{T}\right) x_{t-1} + \sigma_t \eta_t, \nonumber
\end{align}
for $t = 2, \ldots, 12T$, where $T \in \{20, 50, 100\}$ corresponds to 20, 50, and 100 years of monthly data (i.e., 240, 600, and 1200 observations, respectively). The constant term in the predictive regression is set to zero. We set $\beta \in \{0.5k : k = 0, 1, \ldots, 5\}$ and $\bar{\kappa} \in \{0, 50, 100\}$, and consider a one-sided test of $H_0\!: \beta = 0$ against $H_A\!: \beta > 0$. The number of replications is 10,000.

The innovation process $\eta_t$ follows an MA($q$) process:
\begin{align}
\eta_t = \sum_{j=0}^q C_j v_{t-j}, \label{MC3}
\end{align}
where $(\varepsilon_t, v_t)$ are jointly normal with correlation $-0.98$. Specifically, we consider an MA(3) process with $C_j = 1/2$ for $j = 0,1,2,3$ (see Section~5.2 of the main paper for simulations with MA(1) innovations). The volatility process $\sigma_t$ follows the same specifications as in Section~5.2 of the main paper. Specifically, we consider
\begin{itemize}
    \item \textbf{CNST (Constant volatility):} $\sigma_t^2 = \sigma_0^2$, with $\sigma_0 = 1$.
    \item \textbf{SB (Structural break):} $\sigma_t = \sigma_0 + (\sigma_1 - \sigma_0)1\{t/T \ge 4/5\}$, with $\sigma_0 = 1$ and $\sigma_1 = 4$.
    \item \textbf{RS (Regime switching):} $\sigma_t = \sigma_0(1 - s_t) + \sigma_1 s_t$, 
    where $s_t$ is a two-state Markov chain independent of $(y_t, x_t)$, 
    with transition matrix
    \[
    P_t =
    \begin{pmatrix}
    0.8 & 0.2 \\ 0.8 & 0.2
    \end{pmatrix}
    +
    \begin{pmatrix}
    0.2 & -0.2 \\ -0.8 & 0.8
    \end{pmatrix}
    \exp\!\left(-\frac{\bar{\lambda}}{T}t\right),
    \]
    initialized at its invariant distribution, where $\bar{\lambda} = 60$, $\sigma_0 = 1$, and $\sigma_1 = 4$.
\end{itemize}

We implement the hybrid tests based on even and odd observations, denoted by $\tau(\check{\beta}_e)$ and $\tau(\check{\beta}_o)$ (see Section~4.2 of the main paper). For comparison, we also include the IVX test of \citet{KMS2015}, in which the intercept is included in the estimation.

The results, summarized in Tables~\ref{atab1}--\ref{atab3}, indicate that the proposed tests exhibit accurate size control under the null hypothesis across all DGPs, whereas the IVX test is substantially oversized, particularly when volatility is nonstationary or exhibits structural breaks. These findings are consistent with \citet{DEMETRESCU2023105271}, who show that the IVX method exhibits severe size distortions in one-sided tests when regressors are highly persistent and endogenous. By contrast, our tests exhibit strong size control with reasonably good power in the simulation settings considered here.

\begin{table}[h!]
\begin{center}
\scriptsize
\caption{Size and Power for Discrete-Time Models with Constant Volatility and MA(3) Innovations (One-Sided Tests)\label{atab1}}
\begin{tabularx}{0.84\textwidth}{ccccccccccc} \toprule
&&\multicolumn{3}{c}{$\bar{\kappa}=0$}&\multicolumn{3}{c}{$\bar{\kappa}=50$}&\multicolumn{3}{c}{$\bar{\kappa}=100$}\\
\cmidrule(r){3-5}\cmidrule(r){6-8}\cmidrule(r){9-11}
T&&20&50&100&20&50&100&20&50&100\\\hline
$\beta=0$&$\tau(\check{\beta}_o)$&5.0&4.9&4.8&5.2&5.3&4.8&4.5&4.7&5.2\\
&$\tau(\check{\beta}_e)$&4.8&4.7&5.2&5.0&4.8&4.8&5.1&5.0&4.9\\
&IVX&14.2&13.5&12.3&10.2&10.5&10.3&9.5&10.3&10.5\\\hline
$\beta=0.5$&$\tau(\check{\beta}_o)$&7.1&6.7&7.0&12.3&12.5&12.2&16.5&16.6&17.2\\
&$\tau(\check{\beta}_e)$&7.0&7.0&7.2&12.4&11.5&12.3&16.8&17.2&17.5\\
&IVX&100.0&100.0&100.0&100.0&100.0&100.0&100.0&100.0&100.0\\\hline
$\beta=1$&$\tau(\check{\beta}_o)$&10.1&9.8&10.3&25.4&25.8&24.6&39.4&39.4&40.5\\
&$\tau(\check{\beta}_e)$&10.3&10.0&10.4&25.4&24.2&25.4&40.3&40.3&40.8\\
&IVX&100.0&100.0&100.0&100.0&100.0&100.0&100.0&100.0&100.0\\\hline
$\beta=1.5$&$\tau(\check{\beta}_o)$&14.2&14.2&14.7&42.4&44.1&43.1&64.4&65.1&66.4\\
&$\tau(\check{\beta}_e)$&14.4&14.4&14.4&42.8&42.5&42.8&65.9&66.9&66.5\\
&IVX&100.0&100.0&100.0&100.0&100.0&100.0&100.0&100.0&100.0\\\hline
$\beta=2$&$\tau(\check{\beta}_o)$&19.0&18.9&19.6&59.7&61.4&61.6&82.9&84.3&84.6\\
&$\tau(\check{\beta}_e)$&19.4&19.1&19.1&59.5&60.2&60.6&83.4&84.7&84.8\\
&IVX&100.0&100.0&100.0&100.0&100.0&100.0&100.0&100.0&100.0\\\hline
$\beta=2.5$&$\tau(\check{\beta}_o)$&24.1&24.1&24.9&73.1&75.7&75.7&92.5&93.8&94.4\\
&$\tau(\check{\beta}_e)$&24.5&24.0&24.6&73.7&74.3&75.5&92.8&94.1&94.1\\
&IVX&100.0&100.0&100.0&100.0&100.0&100.0&100.0&100.0&100.0\\
 \bottomrule\smallskip
\end{tabularx}
\end{center}
\scriptsize{
This table reports the rejection rates of $\tau(\check{\beta}_o)$ and $\tau(\check{\beta}_e)$, as well as the IVX tests, based on simulated data under $\beta \in \{0, 0.5, 1, 1.5, 2, 2.5\}$  (one-sided tests). The data are generated from the discrete-time model \eqref{MC2}--\eqref{MC3} with persistence parameter $\bar{\kappa} \in \{0,50,100\}$ and MA(3) innovations. The model is sampled at a monthly frequency, so that a sample of length $T$ years contains $12T$ observations, with $T \in \{20,50,100\}$. The volatility is constant.
}
\end{table}

\newpage

\begin{table}[h!]
\begin{center}
\scriptsize
\caption{Size and Power for Discrete-Time Models with Structural Breaks in Volatility and MA(3) Innovations (One-Sided Tests)\label{atab2}}

\begin{tabularx}{0.84\textwidth}{ccccccccccc} \toprule
&&\multicolumn{3}{c}{$\bar{\kappa}=0$}&\multicolumn{3}{c}{$\bar{\kappa}=50$}&\multicolumn{3}{c}{$\bar{\kappa}=100$}\\
\cmidrule(r){3-5}\cmidrule(r){6-8}\cmidrule(r){9-11}
T&&20&50&100&20&50&100&20&50&100\\\hline
$\beta=0$&$\tau(\check{\beta}_o)$&4.6&4.9&5.0&4.5&5.0&5.0&4.3&5.0&5.0\\
&$\tau(\check{\beta}_e)$&4.8&4.7&4.8&4.6&4.6&4.8&4.9&5.1&5.4\\
&IVX&32.7&32.5&31.2&32.1&33.5&34.3&32.4&34.1&35.5\\\hline
$\beta=0.5$&$\tau(\check{\beta}_o)$&9.8&10.1&11.0&7.9&8.9&9.0&9.0&11.4&12.0\\
&$\tau(\check{\beta}_e)$&10.2&10.2&10.4&8.0&8.3&8.8&10.4&11.8&12.4\\
&IVX&100.0&100.0&100.0&100.0&100.0&100.0&100.0&100.0&100.0\\\hline
$\beta=1$&$\tau(\check{\beta}_o)$&18.9&19.2&20.8&13.9&16.0&15.7&18.1&23.3&25.3\\
&$\tau(\check{\beta}_e)$&19.3&20.0&20.0&14.0&14.7&15.8&20.0&24.4&26.2\\
&IVX&100.0&100.0&100.0&100.0&100.0&100.0&100.0&100.0&100.0\\\hline
$\beta=1.5$&$\tau(\check{\beta}_o)$&29.3&29.7&30.9&22.0&25.4&25.5&31.0&39.5&43.1\\
&$\tau(\check{\beta}_e)$&29.4&29.7&30.0&21.9&24.4&25.8&33.9&40.5&43.3\\
&IVX&100.0&100.0&100.0&100.0&100.0&100.0&100.0&100.0&100.0\\\hline
$\beta=2$&$\tau(\check{\beta}_o)$&38.4&38.1&39.8&31.5&36.6&36.8&45.5&56.6&60.8\\
&$\tau(\check{\beta}_e)$&38.5&38.7&38.8&31.8&35.8&37.5&49.0&58.0&60.4\\
&IVX&100.0&100.0&100.0&100.0&100.0&100.0&100.0&100.0&100.0\\\hline
$\beta=2.5$&$\tau(\check{\beta}_o)$&45.9&45.2&46.9&41.8&47.2&48.6&59.1&70.3&74.9\\
&$\tau(\check{\beta}_e)$&45.5&45.7&46.3&41.9&47.2&48.4&62.4&71.8&74.2\\
&IVX&100.0&100.0&100.0&100.0&100.0&100.0&100.0&100.0&100.0\\
 \bottomrule\smallskip
\end{tabularx}
\end{center}
\scriptsize{
This table reports the rejection rates of $\tau(\check{\beta}_o)$ and $\tau(\check{\beta}_e)$, as well as the IVX tests, based on simulated data under $\beta \in \{0, 0.5, 1, 1.5, 2, 2.5\}$  (one-sided tests). The data are generated from the discrete-time model \eqref{MC2}--\eqref{MC3} with persistence parameter $\bar{\kappa} \in \{0,50,100\}$ and MA(3) innovations. The model is sampled at a monthly frequency, so that a sample of length $T$ years contains $12T$ observations, with $T \in \{20,50,100\}$. The volatility exhibits structural breaks.
}
\end{table}

\newpage

\begin{table}[h!]
\begin{center}
\scriptsize
\caption{Size and Power for Discrete-Time Models with  Regime Switching in Volatility with MA(3) Innovations (One-Sided Tests)\label{atab3}}

\begin{tabularx}{0.84\textwidth}{ccccccccccc} \toprule
&&\multicolumn{3}{c}{$\bar{\kappa}=0$}&\multicolumn{3}{c}{$\bar{\kappa}=50$}&\multicolumn{3}{c}{$\bar{\kappa}=100$}\\
\cmidrule(r){3-5}\cmidrule(r){6-8}\cmidrule(r){9-11}
T&&20&50&100&20&50&100&20&50&100\\\hline
$\beta=0$&$\tau(\check{\beta}_o)$&3.8&3.9&4.3&4.3&5.1&5.3&4.6&4.8&4.9\\
&$\tau(\check{\beta}_e)$&4.0&4.2&4.3&4.8&4.9&4.7&4.9&4.8&4.5\\
&IVX&15.6&13.6&13.4&9.8&10.5&11.1&9.6&11.2&12.4\\\hline
$\beta=0.5$&$\tau(\check{\beta}_o)$&7.8&7.9&8.8&10.2&11.6&11.5&12.4&14.5&15.8\\
&$\tau(\check{\beta}_e)$&8.5&8.7&8.7&10.5&11.4&11.9&12.7&14.4&14.9\\
&IVX&100.0&100.0&100.0&100.0&100.0&100.0&100.0&100.0&100.0\\\hline
$\beta=1$&$\tau(\check{\beta}_o)$&15.3&15.4&17.0&22.7&24.9&24.1&27.2&34.7&36.4\\
&$\tau(\check{\beta}_e)$&16.2&16.8&16.9&22.9&24.8&24.0&28.5&34.4&35.7\\
&IVX&100.0&100.0&100.0&100.0&100.0&100.0&100.0&100.0&100.0\\\hline
$\beta=1.5$&$\tau(\check{\beta}_o)$&24.1&24.2&26.8&39.3&42.3&41.9&47.5&58.9&61.5\\
&$\tau(\check{\beta}_e)$&25.4&25.9&26.3&38.3&41.5&41.5&47.6&59.1&60.6\\
&IVX&100.0&100.0&100.0&100.0&100.0&100.0&100.0&100.0&100.0\\\hline
$\beta=2$&$\tau(\check{\beta}_o)$&32.9&33.4&35.4&55.1&59.0&59.1&64.5&77.4&80.3\\
&$\tau(\check{\beta}_e)$&34.3&35.1&35.0&53.1&58.4&58.5&64.7&77.3&79.3\\
&IVX&100.0&100.0&100.0&100.0&100.0&100.0&100.0&100.0&100.0\\\hline
$\beta=2.5$&$\tau(\check{\beta}_o)$&40.9&41.0&42.4&67.3&71.8&72.4&76.2&87.7&91.3\\
&$\tau(\check{\beta}_e)$&42.3&42.4&42.3&64.5&71.5&71.8&76.0&87.6&89.9\\
&IVX&100.0&100.0&100.0&100.0&100.0&100.0&100.0&100.0&100.0\\
 \bottomrule\smallskip
\end{tabularx}
\end{center}
\scriptsize{
This table reports the rejection rates of $\tau(\check{\beta}_o)$ and $\tau(\check{\beta}_e)$, as well as the IVX tests, based on simulated data under $\beta \in \{0, 0.5, 1, 1.5, 2, 2.5\}$  (one-sided tests). The data are generated from the discrete-time model \eqref{MC2}--\eqref{MC3} with persistence parameter $\bar{\kappa} \in \{0,50,100\}$ and MA(3) innovations. The model is sampled at a monthly frequency, so that a sample of length $T$ years contains $12T$ observations, with $T \in \{20,50,100\}$. The volatility exhibits regime switching.
}
\end{table}

\clearpage

\subsection{Discrete-Time Experiments: Two-Sided Tests}

We now consider two-sided tests. The DGP and simulation settings are the same as those for the one-sided tests in Appendix~B.1, but we additionally consider an MA(1) process with $C_0 = C_1 = 1/\sqrt{2}$. For the volatility specifications, we consider only constant volatility and structural breaks to conserve space, as regime-switching volatility yields qualitatively similar results to the structural break case.

The results, summarized in Tables~\ref{atab4}--\ref{atab7}, indicate that the proposed tests exhibit accurate size control under the null hypothesis across all DGPs, whereas the IVX test is substantially oversized when volatility exhibits structural breaks. Our simulation results further show that the IVX approach exhibits severe size distortions not only in one-sided tests but also in two-sided tests when volatility is nonstationary, whereas our methods exhibit accurate size control with reasonably good power in finite samples.

However, we do not argue that our methods should be preferred to the IVX method in the discrete-time predictive regressions. Under constant volatility -- the setting for which the IVX method was originally designed -- the IVX method exhibits substantially higher power than our methods. In this sense, as emphasized in the main paper, our methods and the IVX method of \citet{KMS2015} (as well as the Cauchy RT test of \citet{CJP2016} and the Cauchy VC test of \citet{ibragimov2023new}) should be viewed as complementary rather than competing approaches.

\begin{table}[h!]
\begin{center}
\scriptsize
\caption{Size and Power for Discrete-Time Models with Constant Volatility and MA(1) Innovations (Two-Sided Tests)\label{atab4}}

\begin{tabularx}{0.84\textwidth}{ccccccccccc} \toprule
&&\multicolumn{3}{c}{$\bar{\kappa}=0$}&\multicolumn{3}{c}{$\bar{\kappa}=50$}&\multicolumn{3}{c}{$\bar{\kappa}=100$}\\
\cmidrule(r){3-5}\cmidrule(r){6-8}\cmidrule(r){9-11}
T&&20&50&100&20&50&100&20&50&100\\\hline
$\beta=0$&$\tau(\check{\beta}_o)$&4.9&5.1&5.1&5.2&5.2&4.8&4.7&4.8&5.4\\
&$\tau(\check{\beta}_e)$&5.0&5.1&5.0&5.1&5.2&5.0&5.0&4.5&5.0\\
&IVX&7.1&6.6&5.9&5.2&5.5&5.2&4.6&5.1&5.4\\\hline
$\beta=0.5$&$\tau(\check{\beta}_o)$&6.1&6.0&6.3&10.6&10.9&10.6&16.5&16.5&17.2\\
&$\tau(\check{\beta}_e)$&6.0&5.9&6.1&10.7&10.2&10.6&16.6&16.9&17.1\\
&IVX&100.0&100.0&100.0&100.0&100.0&100.0&100.0&100.0&100.0\\\hline
$\beta=1$&$\tau(\check{\beta}_o)$&9.5&9.0&9.7&29.2&29.8&28.4&49.6&49.8&50.4\\
&$\tau(\check{\beta}_e)$&9.5&9.2&9.5&29.1&27.8&28.6&50.5&50.7&51.4\\
&IVX&100.0&100.0&100.0&100.0&100.0&100.0&100.0&100.0&100.0\\\hline
$\beta=1.5$&$\tau(\check{\beta}_o)$&14.9&14.7&15.3&52.9&54.7&53.9&79.6&80.5&81.0\\
&$\tau(\check{\beta}_e)$&14.8&14.5&14.6&53.4&53.1&53.8&80.0&81.4&81.4\\
&IVX&100.0&100.0&100.0&100.0&100.0&100.0&100.0&100.0&100.0\\\hline
$\beta=2$&$\tau(\check{\beta}_o)$&21.2&21.0&21.8&73.0&75.3&75.1&93.0&94.3&95.1\\
&$\tau(\check{\beta}_e)$&21.1&20.9&21.2&73.0&74.0&74.8&93.5&94.8&95.0\\
&IVX&100.0&100.0&100.0&100.0&100.0&100.0&100.0&100.0&100.0\\\hline
$\beta=2.5$&$\tau(\check{\beta}_o)$&27.5&27.1&28.4&85.3&87.3&87.9&97.7&98.4&98.8\\
&$\tau(\check{\beta}_e)$&27.8&27.5&28.1&85.0&86.7&87.7&98.0&98.5&98.8\\
&IVX&100.0&100.0&100.0&100.0&100.0&100.0&100.0&100.0&100.0\\
 \bottomrule\smallskip
\end{tabularx}
\end{center}
\scriptsize{
This table reports the rejection rates of $\tau(\check{\beta}_o)$ and $\tau(\check{\beta}_e)$, as well as the IVX tests, based on simulated data under $\beta \in \{0, 0.5, 1, 1.5, 2, 2.5\}$  (two-sided tests). The data are generated from the discrete-time model \eqref{MC2}--\eqref{MC3} with persistence parameter $\bar{\kappa} \in \{0,50,100\}$ and MA(1) innovations. The model is sampled at a monthly frequency, so that a sample of length $T$ years contains $12T$ observations, with $T \in \{20,50,100\}$. The volatility is constant.
}
\end{table}

\newpage

\begin{table}[h!]
\begin{center}
\scriptsize
\caption{Size and Power for Discrete-Time Models with Structural Breaks in  Volatility and MA(1) Innovations (Two-Sided Tests)\label{atab5}}

\begin{tabularx}{0.84\textwidth}{ccccccccccc} \toprule
&&\multicolumn{3}{c}{$\bar{\kappa}=0$}&\multicolumn{3}{c}{$\bar{\kappa}=50$}&\multicolumn{3}{c}{$\bar{\kappa}=100$}\\
\cmidrule(r){3-5}\cmidrule(r){6-8}\cmidrule(r){9-11}
T&&20&50&100&20&50&100&20&50&100\\\hline
$\beta=0$&$\tau(\check{\beta}_o)$&4.9&5.1&5.3&5.0&4.9&5.0&5.0&4.8&5.0\\
&$\tau(\check{\beta}_e)$&4.7&5.0&5.3&5.0&5.0&5.1&4.6&4.7&5.1\\
&IVX&20.7&20.9&19.9&24.6&25.8&25.6&25.3&25.5&26.9\\\hline
$\beta=0.5$&$\tau(\check{\beta}_o)$&5.7&6.1&6.2&8.3&8.8&8.2&12.2&11.7&12.3\\
&$\tau(\check{\beta}_e)$&6.1&6.1&6.1&8.2&8.1&7.9&12.1&12.3&12.6\\
&IVX&100.0&100.0&100.0&100.0&100.0&100.0&100.0&100.0&100.0\\\hline
$\beta=1$&$\tau(\check{\beta}_o)$&8.9&9.2&10.1&20.6&19.9&18.9&33.4&33.8&34.3\\
&$\tau(\check{\beta}_e)$&9.4&9.6&9.8&19.7&19.2&18.3&34.5&35.3&34.8\\
&IVX&100.0&100.0&100.0&100.0&100.0&100.0&100.0&100.0&100.0\\\hline
$\beta=1.5$&$\tau(\check{\beta}_o)$&14.5&14.5&15.5&36.8&37.4&36.3&58.8&60.7&62.0\\
&$\tau(\check{\beta}_e)$&15.0&15.0&15.3&36.0&35.5&35.5&61.6&61.9&62.3\\
&IVX&100.0&100.0&100.0&100.0&100.0&100.0&100.0&100.0&100.0\\\hline
$\beta=2$&$\tau(\check{\beta}_o)$&21.3&21.1&21.9&53.0&55.0&54.5&78.2&80.8&81.9\\
&$\tau(\check{\beta}_e)$&21.7&21.3&21.6&53.2&53.2&54.5&79.3&81.3&81.8\\
&IVX&100.0&100.0&100.0&100.0&100.0&100.0&100.0&100.0&100.0\\\hline
$\beta=2.5$&$\tau(\check{\beta}_o)$&27.5&27.2&28.8&66.4&69.5&69.0&88.3&91.4&92.4\\
&$\tau(\check{\beta}_e)$&28.2&27.7&27.8&66.8&67.6&69.1&89.5&91.4&92.1\\
&IVX&100.0&100.0&100.0&100.0&100.0&100.0&100.0&100.0&100.0\\
 \bottomrule\smallskip
\end{tabularx}
\end{center}
\scriptsize{
This table reports the rejection rates of $\tau(\check{\beta}_o)$ and $\tau(\check{\beta}_e)$, as well as the IVX tests, based on simulated data under $\beta \in \{0, 0.5, 1, 1.5, 2, 2.5\}$  (two-sided tests). The data are generated from the discrete-time model \eqref{MC2}--\eqref{MC3} with persistence parameter $\bar{\kappa} \in \{0,50,100\}$ and MA(1) innovations. The model is sampled at a monthly frequency, so that a sample of length $T$ years contains $12T$ observations, with $T \in \{20,50,100\}$. The volatility exhibits structural breaks.
}
\end{table}

\newpage

\begin{table}[h!]
\begin{center}
\scriptsize
\caption{Size and Power for Discrete-Time Models with Constant Volatility and MA(3) Innovations (Two-Sided Tests)\label{atab6}}

\begin{tabularx}{0.84\textwidth}{ccccccccccc} \toprule
&&\multicolumn{3}{c}{$\bar{\kappa}=0$}&\multicolumn{3}{c}{$\bar{\kappa}=50$}&\multicolumn{3}{c}{$\bar{\kappa}=100$}\\
\cmidrule(r){3-5}\cmidrule(r){6-8}\cmidrule(r){9-11}
T&&20&50&100&20&50&100&20&50&100\\\hline
$\beta=0$&$\tau(\check{\beta}_o)$&5.0&4.8&5.1&4.8&4.9&5.0&4.8&4.8&5.1\\
&$\tau(\check{\beta}_e)$&4.8&5.1&5.0&5.0&4.9&4.9&4.6&4.9&5.3\\
&IVX&8.0&7.0&6.3&4.9&5.3&5.3&4.8&4.8&5.5\\\hline
$\beta=0.5$&$\tau(\check{\beta}_o)$&7.5&7.2&7.4&6.6&7.4&6.9&6.8&9.0&9.7\\
&$\tau(\check{\beta}_e)$&7.2&7.2&7.4&6.8&6.6&7.3&7.6&9.4&10.0\\
&IVX&100.0&100.0&100.0&100.0&100.0&100.0&100.0&100.0&100.0\\\hline
$\beta=1$&$\tau(\check{\beta}_o)$&13.9&14.0&14.6&13.1&14.7&14.6&17.6&23.1&25.4\\
&$\tau(\check{\beta}_e)$&14.0&14.1&14.4&12.8&13.2&15.3&17.8&23.4&25.5\\
&IVX&100.0&100.0&100.0&100.0&100.0&100.0&100.0&100.0&100.0\\\hline
$\beta=1.5$&$\tau(\check{\beta}_o)$&22.7&22.8&24.1&23.8&27.5&27.8&35.6&45.8&49.6\\
&$\tau(\check{\beta}_e)$&22.8&23.4&23.8&23.5&25.6&28.2&35.9&46.3&48.7\\
&IVX&100.0&100.0&100.0&100.0&100.0&100.0&100.0&100.0&100.0\\\hline
$\beta=2$&$\tau(\check{\beta}_o)$&31.3&31.4&32.8&38.2&43.6&44.7&55.8&68.0&71.9\\
&$\tau(\check{\beta}_e)$&31.6&31.9&33.0&37.8&41.1&44.3&56.0&67.8&71.0\\
&IVX&100.0&100.0&100.0&100.0&100.0&100.0&100.0&100.0&100.0\\\hline
$\beta=2.5$&$\tau(\check{\beta}_o)$&38.5&38.8&40.1&52.3&59.1&60.7&72.8&83.5&86.3\\
&$\tau(\check{\beta}_e)$&38.7&39.3&40.4&52.0&56.9&60.5&72.4&83.1&86.0\\
&IVX&100.0&100.0&100.0&100.0&100.0&100.0&100.0&100.0&100.0\\
 \bottomrule\smallskip
\end{tabularx}
\end{center}
\scriptsize{
This table reports the rejection rates of $\tau(\check{\beta}_o)$ and $\tau(\check{\beta}_e)$, as well as the IVX tests, based on simulated data under $\beta \in \{0, 0.5, 1, 1.5, 2, 2.5\}$  (two-sided tests). The data are generated from the discrete-time model \eqref{MC2}--\eqref{MC3} with persistence parameter $\bar{\kappa} \in \{0,50,100\}$ and MA(3) innovations. The model is sampled at a monthly frequency, so that a sample of length $T$ years contains $12T$ observations, with $T \in \{20,50,100\}$. The volatility is constant.
}
\end{table}

\newpage

\begin{table}[h!]
\begin{center}
\scriptsize
\caption{Size and Power for Discrete-Time Models with Structural Breaks in Volatility and MA(3) Innovations (Two-Sided Tests)\label{atab7}}

\begin{tabularx}{0.84\textwidth}{ccccccccccc} \toprule
&&\multicolumn{3}{c}{$\bar{\kappa}=0$}&\multicolumn{3}{c}{$\bar{\kappa}=50$}&\multicolumn{3}{c}{$\bar{\kappa}=100$}\\
\cmidrule(r){3-5}\cmidrule(r){6-8}\cmidrule(r){9-11}
T&&20&50&100&20&50&100&20&50&100\\\hline
$\beta=0$&$\tau(\check{\beta}_o)$&4.8&5.1&5.2&5.0&5.0&5.0&4.5&4.7&5.1\\
&$\tau(\check{\beta}_e)$&4.9&5.0&5.2&5.2&5.1&4.7&5.0&4.8&5.4\\
&IVX&22.0&22.1&20.5&24.2&24.9&25.8&24.4&25.5&27.0\\\hline
$\beta=0.5$&$\tau(\check{\beta}_o)$&6.8&7.3&7.4&5.5&5.9&6.2&5.6&7.3&7.5\\
&$\tau(\check{\beta}_e)$&7.1&7.0&7.4&5.7&5.4&5.8&6.3&7.2&8.0\\
&IVX&100.0&100.0&100.0&100.0&100.0&100.0&100.0&100.0&100.0\\\hline
$\beta=1$&$\tau(\check{\beta}_o)$&13.4&13.7&15.0&8.7&9.9&9.8&11.0&15.0&16.6\\
&$\tau(\check{\beta}_e)$&14.0&14.1&14.3&8.7&9.1&9.9&12.5&15.6&17.1\\
&IVX&100.0&100.0&100.0&100.0&100.0&100.0&100.0&100.0&100.0\\\hline
$\beta=1.5$&$\tau(\check{\beta}_o)$&22.4&22.9&24.6&14.9&17.0&17.0&21.7&29.2&31.9\\
&$\tau(\check{\beta}_e)$&23.2&23.6&23.6&15.1&16.0&17.3&23.4&29.7&32.6\\
&IVX&100.0&100.0&100.0&100.0&100.0&100.0&100.0&100.0&100.0\\\hline
$\beta=2$&$\tau(\check{\beta}_o)$&31.3&31.6&33.2&22.9&26.9&27.0&34.8&44.9&49.5\\
&$\tau(\check{\beta}_e)$&31.6&32.0&32.4&23.0&25.6&27.3&37.6&46.3&49.6\\
&IVX&100.0&100.0&100.0&100.0&100.0&100.0&100.0&100.0&100.0\\\hline
$\beta=2.5$&$\tau(\check{\beta}_o)$&38.5&38.7&40.3&32.0&37.3&37.7&48.6&60.5&65.6\\
&$\tau(\check{\beta}_e)$&39.1&39.5&39.9&32.3&36.2&38.5&51.6&62.0&65.2\\
&IVX&100.0&100.0&100.0&100.0&100.0&100.0&100.0&100.0&100.0\\
 \bottomrule\smallskip
\end{tabularx}
\end{center}
\scriptsize{
This table reports the rejection rates of $\tau(\check{\beta}_o)$ and $\tau(\check{\beta}_e)$, as well as the IVX tests, based on simulated data under $\beta \in \{0, 0.5, 1, 1.5, 2, 2.5\}$  (two-sided tests). The data are generated from the discrete-time model \eqref{MC2}--\eqref{MC3} with persistence parameter $\bar{\kappa} \in \{0,50,100\}$ and MA(3) innovations. The model is sampled at a monthly frequency, so that a sample of length $T$ years contains $12T$ observations, with $T \in \{20,50,100\}$. The volatility exhibits structural breaks.
}
\end{table}

\clearpage
\subsection{Simulation Results for Joint Predictability Tests}

In this subsection, we examine the finite-sample performance of the proposed Wald test in Appendix~A.1 and compare it with the IVX test. Specifically, we consider a joint predictability testing problem with two predictors. The DGP is specified as
\begin{align}
y_t &= \beta_1 x_{1,t-1} + \beta_2 x_{2,t-1} + \sigma_t \varepsilon_t, \label{MC4} \\
x_{it} &= \left(1 - \frac{\bar{\kappa}}{T}\right) x_{i,t-1} + \sigma_t \eta_{it}, \quad i = 1,2, \nonumber
\end{align}
for $t = 2, \ldots, 12T$, where $T \in \{20, 50, 100\}$ corresponds to 20, 50, and 100 years of monthly data (i.e., 240, 600, and 1200 observations, respectively). The constant term in the predictive regression is set to zero. We set $\beta_1 = \beta_2$ with $\beta_1 \in \{0.05k : k = 0, 1, \ldots, 5\}$ and $\bar{\kappa} \in \{0, 50, 100\}$, and consider the joint test of $H_0\!: \beta_1 = \beta_2 = 0$ against $H_A\!: \beta_1 \neq 0 \text{ or } \beta_2 \neq 0$.

The innovation process $\eta_{it}$ follows an MA($q$) process:
\begin{align}
\eta_{it} = \sum_{j=0}^q C_j v_{i,t-j}, \label{MC5}
\end{align}
where $(\varepsilon_t, v_{1t}, v_{2t})$ are jointly normal with the correlation matrix
\[
\begin{pmatrix}
1.0000 & -0.7596 & -0.2787 \\
-0.7596 & 1.0000 & 0.1246 \\
-0.2787 & 0.1246 & 1.0000
\end{pmatrix}.
\]
We consider two specifications for the innovation process: MA(1) with $C_0 = C_1 = 1/\sqrt{2}$ and MA(3) with $C_0 = C_1 = C_2 = C_3 = 1/2$. The number of replications is 1,000.

The volatility process $\sigma_t$ follows the same specifications as in Appendix~B.1. The intercept is set to zero. The simulation design focuses on the case without an intercept in order to isolate the performance of the joint testing procedure from issues arising in predictive regressions with a nonzero intercept. While it is, in principle, possible to combine the Wald-type test with the first-differencing approach to accommodate a nonzero intercept, such an extension introduces an additional robustness--efficiency trade-off, as first differencing generally reduces power, as discussed in Appendix~A.2.

Tables~\ref{tab5} and~\ref{tab6} report the results for the Wald test in Appendix A.1 and the IVX-Wald test of \citet{KMS2015} under MA(1) and MA(3) innovations, respectively. For the mean-reverting designs ($\bar{\kappa}=50,100$) the Wald test controls size well across all volatility specifications and, as in the tests for individual parameters in Sections~B.1 and~B.2, retains good power; it is also far less distorted than the IVX-Wald test, which is severely oversized under structural breaks and regime switching. In the exact unit-root case with endogenous regressors ($\bar{\kappa}=0$), the Wald test shows a mild over-rejection -- rejection rates of about $7$--$10\%$. This is precisely the regime in which the asymptotic-exogeneity condition of Assumption~\ref{ass-a1}(v-b) is not expected to hold, and it reflects the robustness trade-off discussed in Appendix~A.1: the component-wise Bonferroni combination of the univariate Cauchy tests, which does not require this condition, remains the robust default for joint predictability testing when the regressors are both persistent and endogenous.

\newpage

\begin{table}[h!]
\begin{center}
\scriptsize
\caption{Size and Power for Discrete-Time Models with MA(1) Innovations (Joint Tests)\label{tab5}}
\begin{tabularx}{0.84\textwidth}{ccccccccccc} \toprule
&&\multicolumn{3}{c}{$\bar{\kappa}=0$}&\multicolumn{3}{c}{$\bar{\kappa}=50$}&\multicolumn{3}{c}{$\bar{\kappa}=100$}\\
\cmidrule(r){3-5}\cmidrule(r){6-8}\cmidrule(r){9-11}
T&&20&50&100&20&50&100&20&50&100\\\hline
\multicolumn{11}{c}{Constant Volatility}\\\hline
$\beta=0$&Wald&6.6&8.3&5.5&6.4&4.6&5.4&5.6&5.4&5.6\\
&IVX&7.6&7.0&6.2&4.8&6.2&5.6&6.0&4.8&5.4\\
$\beta=0.05$&Wald&98.3&99.6&99.8&44.7&100.0&100.0&25.8&91.9&100.0\\
&IVX&100.0&100.0&100.0&64.5&100.0&100.0&38.4&99.2&100.0\\
$\beta=0.1$&Wald&99.4&99.7&100.0&97.4&100.0&100.0&76.4&100.0&100.0\\
&IVX&100.0&100.0&100.0&99.8&100.0&100.0&93.2&100.0&100.0\\
$\beta=0.15$&Wald&99.6&99.7&100.0&100.0&100.0&100.0&98.6&100.0&100.0\\
&IVX&100.0&100.0&100.0&100.0&100.0&100.0&99.9&100.0&100.0\\
$\beta=0.2$&Wald&99.6&99.8&100.0&100.0&100.0&100.0&100.0&100.0&100.0\\
&IVX&100.0&100.0&100.0&100.0&100.0&100.0&100.0&100.0&100.0\\
\hline
\multicolumn{11}{c}{Structural Breaks in Volatility}\\\hline
$\beta=0$&Wald&7.5&8.6&9.0&6.8&5.5&5.8&6.2&5.5&4.7\\
&IVX&26.9&31.8&28.2&38.6&38.9&41.7&40.0&36.7&39.3\\
$\beta=0.05$&Wald&97.5&99.7&100.0&30.5&97.4&100.0&17.2&74.5&100.0\\
&IVX&99.7&100.0&100.0&69.3&99.7&100.0&56.4&93.9&100.0\\
$\beta=0.1$&Wald&99.2&99.8&100.0&84.4&100.0&100.0&56.0&100.0&100.0\\
&IVX&100.0&100.0&100.0&96.1&100.0&100.0&84.2&100.0&100.0\\
$\beta=0.15$&Wald&99.4&99.8&100.0&98.9&100.0&100.0&87.2&100.0&100.0\\
&IVX&100.0&100.0&100.0&100.0&100.0&100.0&97.3&100.0&100.0\\
$\beta=0.2$&Wald&99.7&99.9&100.0&100.0&100.0&100.0&99.4&100.0&100.0\\
&IVX&100.0&100.0&100.0&100.0&100.0&100.0&99.8&100.0&100.0\\
\hline
\multicolumn{11}{c}{Regime Switching in Volatility}\\\hline
$\beta=0$&Wald&9.4&9.8&7.8&6.2&5.3&5.2&5.6&4.7&5.5\\
&IVX&12.3&12.2&10.8&15.6&16.7&15.4&17.1&16.8&19.0\\
$\beta=0.05$&Wald&98.0&99.7&100.0&39.9&99.1&100.0&24.8&85.7&100.0\\
&IVX&99.9&100.0&100.0&68.5&99.9&100.0&44.7&97.8&100.0\\
$\beta=0.1$&Wald&99.4&99.9&100.0&93.4&100.0&100.0&67.9&100.0&100.0\\
&IVX&100.0&100.0&100.0&99.1&100.0&100.0&89.3&100.0&100.0\\
$\beta=0.15$&Wald&99.7&99.9&100.0&100.0&100.0&100.0&96.5&100.0&100.0\\
&IVX&100.0&100.0&100.0&100.0&100.0&100.0&99.1&100.0&100.0\\
$\beta=0.2$&Wald&99.7&99.9&100.0&100.0&100.0&100.0&99.9&100.0&100.0\\
&IVX&100.0&100.0&100.0&100.0&100.0&100.0&100.0&100.0&100.0\\
 \bottomrule\smallskip
\end{tabularx}
\end{center}
\scriptsize{
This table reports the rejection rates of the Wald test in Appendix A.1 and the IVX-Wald test of \citet{KMS2015}, based on simulated data under $\beta \in \{0, 0.05, 0.1, 0.15, 0.2\}$  (joint tests). The data are generated from the discrete-time model \eqref{MC4}--\eqref{MC5} with persistence parameter $\bar{\kappa} \in \{0,50,100\}$ and MA(1) innovations. The model is sampled at a monthly frequency, so that a sample of length $T$ years contains $12T$ observations, with $T \in \{20,50,100\}$. 
}
\end{table}

\newpage

\begin{table}[h!]
\begin{center}
\scriptsize
\caption{Size and Power for Discrete-Time Models with MA(3) Innovations (Joint Tests)\label{tab6}}
\begin{tabularx}{0.84\textwidth}{ccccccccccc} \toprule
&&\multicolumn{3}{c}{$\bar{\kappa}=0$}&\multicolumn{3}{c}{$\bar{\kappa}=50$}&\multicolumn{3}{c}{$\bar{\kappa}=100$}\\
\cmidrule(r){3-5}\cmidrule(r){6-8}\cmidrule(r){9-11}
T&&20&50&100&20&50&100&20&50&100\\\hline
\multicolumn{11}{c}{Constant Volatility}\\\hline
$\beta=0$&Wald&6.8&7.9&6.4&5.9&3.9&4.7&5.7&6&6.2\\
&IVX&8&7.2&6.6&4.9&5.6&5.7&6.1&5.9&5.4\\
$\beta=0.05$&Wald&98.9&99.7&100&67.3&100&100&34.7&99.7&100\\
&IVX&100&100&100&88.1&100&100&49&100&100\\
$\beta=0.1$&Wald&99.3&99.8&100&99.8&100&100&91.4&100&100\\
&IVX&100&100&100&100&100&100&98.3&100&100\\
$\beta=0.15$&Wald&99.4&99.8&100&100&100&100&99.8&100&100\\
&IVX&100&100&100&100&100&100&100&100&100\\
$\beta=0.2$&Wald&99.5&99.8&100&100&100&100&100&100&100\\
&IVX&100&100&100&100&100&100&100&100&100\\
\hline
\multicolumn{11}{c}{Structural Breaks in Volatility}\\\hline
$\beta=0$&Wald&8.2&9.0&8.6&5.8&5.7&5.7&6.6&6.6&5.6\\
&IVX&28.5&32.2&27.9&34.7&38.6&39.9&41.6&37.8&39.0\\
$\beta=0.05$&Wald&98.9&99.6&100.0&44.7&99.9&100.0&23.2&93.7&100.0\\
&IVX&100.0&100.0&100.0&79.3&100.0&100.0&60.3&99.3&100.0\\
$\beta=0.1$&Wald&99.6&99.7&100.0&96.1&100.0&100.0&72.5&100.0&100.0\\
&IVX&100.0&100.0&100.0&99.6&100.0&100.0&91.9&100.0&100.0\\
$\beta=0.15$&Wald&99.6&99.8&100.0&100.0&100.0&100.0&97.0&100.0&100.0\\
&IVX&100.0&100.0&100.0&100.0&100.0&100.0&99.7&100.0&100.0\\
$\beta=0.2$&Wald&99.7&99.8&100.0&100.0&100.0&100.0&99.8&100.0&100.0\\
&IVX&100.0&100.0&100.0&100.0&100.0&100.0&100.0&100.0&100.0\\
\hline
\multicolumn{11}{c}{Regime Switching in Volatility}\\\hline
$\beta=0$&Wald&10.2&9.3&8.7&5.8&4.4&4.8&5.7&5.8&5.5\\
&IVX&12.8&12.5&11.5&13.2&16.1&16.1&17.1&16.5&19.7\\
$\beta=0.05$&Wald&98.6&99.7&100.0&60.7&100.0&100.0&30.1&98.7&100.0\\
&IVX&100.0&100.0&100.0&86.1&100.0&100.0&51.7&99.7&100.0\\
$\beta=0.1$&Wald&99.5&99.7&100.0&99.6&100.0&100.0&84.2&100.0&100.0\\
&IVX&100.0&100.0&100.0&99.9&100.0&100.0&96.0&100.0&100.0\\
$\beta=0.15$&Wald&99.6&99.8&100.0&100.0&100.0&100.0&99.2&100.0&100.0\\
&IVX&100.0&100.0&100.0&100.0&100.0&100.0&99.9&100.0&100.0\\
$\beta=0.2$&Wald&99.6&99.8&100.0&100.0&100.0&100.0&99.9&100.0&100.0\\
&IVX&100.0&100.0&100.0&100.0&100.0&100.0&100.0&100.0&100.0\\
 \bottomrule\smallskip
\end{tabularx}
\end{center}
\scriptsize{
This table reports the rejection rates of the Wald test in Appendix A.1 and the IVX-Wald test of \citet{KMS2015}, based on simulated data under $\beta \in \{0, 0.05, 0.1, 0.15, 0.2\}$  (joint tests). The data are generated from the discrete-time model \eqref{MC4}--\eqref{MC5} with persistence parameter $\bar{\kappa} \in \{0,50,100\}$ and MA(3) innovations. The model is sampled at a monthly frequency, so that a sample of length $T$ years contains $12T$ observations, with $T \in \{20,50,100\}$. 
}
\end{table}

\newpage

\section{Additional Empirical Results}

\subsection{Two-Sided Tests with Intercept}

In this subsection, we report empirical results for tests that account for the intercept. We use the same dataset as in Section~6 of the main paper. Table~\ref{ctab1} presents the results. In contrast to Table~9 in the main text, accounting for the intercept may affect the conclusions. In particular, the methods yield different outcomes: (i) the hybrid method with first differencing tends to support no predictability (\textcolor{black}{one rejection out of 18 tests for each of $\tau(\check{\beta}_o)$ and $\tau(\check{\beta}_e)$}), and (ii) the IVX method rejects the null hypothesis more frequently than the hybrid method (\textcolor{black}{8} rejections out of 18 tests). These mixed results may reflect the lower power of the hybrid method with first differencing or size distortions of the IVX method. As shown in Appendix~B.2, the IVX method exhibits size distortions under structural breaks in volatility, even for two-sided tests, whereas the hybrid method exhibits good power properties, albeit with reduced power, especially when $\beta$ is close to zero.

\begin{table}[h!]
\begin{center}
\footnotesize
\caption{Empirical Results on Stock Return Predictability (Two-Sided Tests)\label{ctab1}}

\begin{tabularx}{0.74\textwidth}{lllll} \toprule
Series &Frequency&$\tau(\check{\beta}_o)$	&$\tau(\check{\beta}_e)$ &	IVX\\
\hline
\multicolumn{5}{l}{Panel A: D/P as predictor for the period of 1927-2011}\\\hline
CRSP&Monthly&0.51&0.51&\textcolor{black}{1.95}\\
&Quarterly&0.71&0.67&2.52$^{*}$\\
&Yearly&0.91&0.73&2.55$^{*}$\\
S\&P500&Monthly&0.68&0.68&\textcolor{black}{2.01$^{*}$}\\
&Quarterly&0.77&0.74&2.73$^{**}$\\
&Yearly&0.68&0.86&1.93\\
\hline
\multicolumn{5}{l}{Panel B: D/P as predictor for the period of 1927-2011 with jumps removed}\\\hline
CRSP&Monthly&-0.16&0.15&2.11$^{*}$\\
&Quarterly&0.22&-0.27&2.79$^{**}$\\
&Yearly&0.57&0.45&1.42\\
S\&P500&Monthly&1.04&1.05&\textcolor{black}{2.10$^{*}$}\\
&Quarterly&2.17$^{*}$&2.21$^{*}$&2.47$^{*}$\\
&Yearly&-0.22&-0.01&1.77\\
\hline
\multicolumn{5}{l}{Panel C: E/P as predictor for the period of 1950-2011}\\\hline
S\&P500&Monthly&0.35&-0.36&0.98\\
&Quarterly&\textcolor{black}{1.65}&1.64&1.10\\
&Yearly&1.14&1.21&1.17\\
\hline
\multicolumn{5}{l}{Panel D: E/P as predictor for the period of 1950-2011 with jumps removed}\\\hline
S\&P500&Monthly&0.88&0.88&0.76\\
&Quarterly&-0.14&0.08&0.93\\
&Yearly&0.72&0.50&1.08\\
 \bottomrule\smallskip
\end{tabularx}
\end{center}
\scriptsize{Results for two-sided tests of return predictability for the NYSE/AMEX value-weighted index (CRSP) and the S\&P~500 using three intercept-robust tests, \(\tau(\check{\beta}_o)\), \(\tau(\check{\beta}_e)\), and the IVX method, across different regression frequencies. Panels~A--B use the dividend--price ratio (D/P), while Panels~C--D use the earnings--price ratio (E/P) as predictors. Statistical significance at the 5\% and 1\% levels is denoted by ``$^{*}$'' and ``$^{**}$'', respectively.}
\end{table}

\clearpage

\subsection{Joint Predictability: Dividend--Price and Earnings--Price Ratios}

We now consider the joint predictability of the dividend--price ratio (D/P) and the earnings--price ratio (E/P) for stock returns. We use S\&P~500 returns, as both D/P and E/P are readily available for this index. We implement two tests: the Wald-type test in Appendix~A.1 and the IVX-Wald test of \citet{KMS2015}, and test the null hypothesis of no predictability jointly. Table~\ref{ctab2} presents the results for monthly, quarterly, and yearly S\&P~500 series using both the earnings--price and dividend--price ratios over the period 1950--2011. Recursive demeaning is employed as in \citet{CJP2016}. In all cases, comparing with the corresponding $\chi^2$ critical values at the 1\% significance level, we fail to reject the null hypothesis of no predictability.

\begin{table}[h!]
\begin{center}
\footnotesize
\caption{Empirical Results on Stock Return Predictability (Joint Tests)\label{ctab2}}
\begin{tabularx}{0.50\textwidth}{lllll} \toprule
&Wald&IVX&Wald&IVX\\\hline
&\multicolumn{2}{l}{With jumps}&\multicolumn{2}{l}{With jumps removed}\\\hline
Monthly	&1.47&	1.43&	0.72&	1.13\\
Quarterly&	0.54&	2.47&	1.40&	0.88\\
Yearly&	2.02&	2.35&	1.95&	0.42\\
 \bottomrule\smallskip
\end{tabularx}
\end{center}
\scriptsize{Results for joint tests of return predictability for the S\&P~500 using the Wald test in Appendix~A.1, and the IVX-Wald test of \citet{KMS2015}, across different regression frequencies. The two predictors are the dividend--price ratio and the earnings--price ratio.}
\end{table}

\section{Proofs of the Main Results}\label{proofs-main}

This appendix collects the proofs of the results stated in the main paper. Notation, together with all equation, assumption, lemma, proposition, and corollary numbers, follows the main paper.

\begin{proof}[Proof of Lemma~\ref{lemma-2-1}]
For $j = 1, \dots, q$, we have
\[
\check{\gamma}_j - \zeta_j = \sqrt{\frac{q}{T}} \sum_{t = (j-1)[T/q] + 1}^{j[T/q]} \operatorname{sign}(x_{t-1}) u_t 
= \sqrt{q} \int_{(j-1)/q}^{j/q} \sigma^T(r) \, dW^T(r).
\]
The stated result follows immediately from Assumptions~\ref{assumption-1-2} and~\ref{assumption-2-1}.
\end{proof}

\begin{proof}[Proof of Proposition~\ref{proposition-2-1}]
Part~(a) follows from Theorem~1 and the discussion in Section~2.2 of \citet{IbragimovMuller2010}.  
For part~(b), we deduce from Lemma~\ref{lemma-2-1} that
\begin{align}\label{proof-eq1}
\frac{1}{\sqrt{T/q}} \check{\gamma}_j 
= \frac{1}{\sqrt{T/q}} \zeta_j + o_p\!\left(\frac{1}{\sqrt{T}}\right)
\to_p \beta \,  E|x_t|    
\end{align}
uniformly in~$j$, under $\beta \neq 0$.  
Recall that
\[
t_q(\check{\gamma}) = \frac{\sqrt{q}\,\overline{\gamma}}{s_{\gamma}},
\quad
\text{with } 
\overline{\gamma} = \frac{1}{q}\sum_{j=1}^q \check{\gamma}_j, 
\quad
s_{\gamma}^2 = \frac{1}{q-1} \sum_{j=1}^q (\check{\gamma}_j - \overline{\gamma})^2.
\]
Hence, the numerator of $t_q(\check{\gamma})$ satisfies
\begin{align}\label{proof-eq2}
\frac{\sqrt{q}\,\overline{\gamma}}{\sqrt{T}} \to_p \beta \,  E|x_t|.
\end{align}
To complete the proof, it suffices to show that $s_{\gamma}^2 = o_p(T)$. Indeed, for $q \ge 2$, we have
\[
\frac{q(q-1)}{T} s_{\gamma}^2 
= \frac{1}{T/q}\sum_{j=1}^q (\check{\gamma}_j - \overline{\gamma})^2 
\to_p 0
\]
due to \eqref{proof-eq1} and \eqref{proof-eq2}, which completes the proof.
\end{proof}

\begin{proof}[Proof of Corollary~\ref{corollary-2-1}]
We aim to show that
\begin{align}\label{proof-eq3}
\frac{\sqrt{q}}{T} \overline{\gamma} \to_d \beta \int_0^1 |X(r)|\,dr,
\end{align}
and
\begin{align}\label{proof-eq4}
\frac{q(q-1)}{T^2} s_{\gamma}^2 
= \frac{q}{T^2}\sum_{j=1}^q (\check{\gamma}_j - \overline{\gamma})^2
\to_d \beta^2 \sum_{j=1}^q 
\left( q \int_{(j-1)/q}^{j/q} |X(r)|\,dr - \int_0^1 |X(r)|\,dr \right)^2.
\end{align}

For \eqref{proof-eq3}, we have
\begin{align*}
\frac{\sqrt{q}}{T} \check{\gamma}_j 
&= \frac{\sqrt{q}}{T} \zeta_j + o_p\!\left(\frac{1}{T}\right)
= \beta \frac{q}{T^{3/2}} \sum_{t = (j-1)[T/q] + 1}^{j[T/q]} |x_{t-1}| + o_p(1)
\to_d \beta q \int_{(j-1)/q}^{j/q} |X(r)|\,dr,
\end{align*}
by Lemma~\ref{lemma-2-1} and~\eqref{fclt}, leading to \eqref{proof-eq3}.  
Moreover,
\[
\frac{q}{T^2}(\check{\gamma}_j - \overline{\gamma})^2 
\to_d \beta^2 \left( q \int_{(j-1)/q}^{j/q} |X(r)|\,dr - \int_0^1 |X(r)|\,dr \right)^2,
\]
which yields \eqref{proof-eq4}.

Combining \eqref{proof-eq3} and \eqref{proof-eq4}, we obtain
\[
t_q(\check{\gamma}) = \frac{\sqrt{q}\,\overline{\gamma}}{s_\gamma} \to_d  \operatorname{sign}(\beta)\mathcal{D}_q,
\]
where 
\[
\mathcal{D}_q 
= \int_0^1 |X(r)|\,dr 
\left( 
\frac{q(q-1)}{\sum_{j=1}^q 
(\int_0^1 |X(r)|\,dr - q\int_{(j-1)/q}^{j/q} |X(r)|\,dr)^2}
\right)^{1/2},
\]
as desired.  
Furthermore, $(q-1)^{1/2}\mathcal{D}_q > 1$ for $q \ge 2$ with probability one, since
\begin{align*}
(q-1)^{1/2}\mathcal{D}_q
&= (q-1) \int_0^1 |X(r)|\,dr 
\left(
\frac{q}{\sum_{j=1}^q (\int_0^1 |X(r)|\,dr - q\int_{(j-1)/q}^{j/q} |X(r)|\,dr)^2}
\right)^{1/2} \\
&> \frac{(q-1)\int_0^1 |X(r)|\,dr}{
\max_{1 \le j \le q}
\left| \int_0^1 |X(r)|\,dr - q\int_{(j-1)/q}^{j/q} |X(r)|\,dr \right|
} > 1,
\end{align*}
which completes the proof.
\end{proof}

\begin{proof}[Proof of Proposition~\ref{proposition-2-2}]
It suffices to show that $\hat{\omega}^2 \to_p \omega^2$, since this implies
\[
\tau(\check{\beta}) 
= \frac{\check{\gamma}}{\omega}(1 + o_p(1))
= \left( \beta \frac{\sum_{t=1}^T |x_{t-1}|}{\omega\sqrt{T}} 
+ \frac{\sum_{t=1}^T \operatorname{sign}(x_{t-1})u_t}{\omega\sqrt{T}} \right)(1 + o_p(1)),
\]
where, in particular,
\[
\frac{\sum_{t=1}^T \operatorname{sign}(x_{t-1})u_t}{\omega\sqrt{T}} \to_d \mathbb{N}(0,1)
\]
under Assumptions~\ref{assumption-mds}, \ref{assumption-1-2}, and~\ref{assumption-2-1}.  

Let $\omega_T^2 = T^{-1} \sum_{t=1}^T u_t^2$.  
Then $\omega_T^2 \to_p \omega^2$ by Assumptions~\ref{assumption-mds} and~\ref{assumption-1-2}.  
Furthermore, 
\[
\hat{\omega}^2 
= \omega_T^2 - \frac{1}{T} \frac{(\sum_{t=1}^T x_{t-1}u_t)^2}{\sum_{t=1}^T x_{t-1}^2} 
= \omega_T^2 + o_p(1)
\]
by Assumption~\ref{assumption-ols}, which gives the desired result.
\end{proof}

\begin{proof}[Proof of Lemma~\ref{lemma-4-1}]
We need only show that $H_0$ holds if $H_0^{(k)}$ holds for all~$k$.  
Let $C = [\beta_1, \ldots, \beta_K]'$.  
By the moment restrictions, $C$ is the solution to
\[
 E[z_{t-1} y_t] 
=  E\!\left[
\operatorname{diag}(|x_{1,t-1}|, \ldots, |x_{K,t-1}|)
\right] C.
\]
Given $0 <  E[|x_{k,t-1}|] < \infty$ for all~$k$, $H_0^{(k)}$ holds for all~$k$ if and only if $ E[z_{t-1} y_t] = \mathbf{0}_{K \times 1}$.

Moreover,
\[
 E[z_{t-1} y_t] =  E[z_{t-1} x_{t-1}'] B,
\]
and since $ E[z_{t-1} x_{t-1}']$ is assumed invertible, the condition $ E[z_{t-1} y_t] = \mathbf{0}_{K \times 1}$ implies $B = 0$, completing the proof.
\end{proof}

\begin{proof}[Proof of Corollary~\ref{corollary-3-2}]
Under the modified version of Assumption~\ref{assumption-ols} stated in the corollary, it follows from arguments analogous to those in the proof of Proposition~\ref{proposition-2-2} that $\hat{\omega}^2 \to_p \omega^2$.  
Hence, it suffices to show that
\[
\frac{1}{\omega\sqrt{T}} \sum_{t=1}^{T/2} 
\operatorname{sign}(x_{2t-2}) (u_{2t} - u_{2t-1}) \to_d \mathbb{N}(0,1).
\]

Define $\xi_t = \operatorname{sign}(x_{t-2})u_t$ for even~$t$ and $\xi_t = -\operatorname{sign}(x_{t-1})u_t$ for odd~$t$, so that
\[
\frac{1}{\omega\sqrt{T}} \sum_{t=1}^{T/2} \operatorname{sign}(x_{2t-2})(u_{2t} - u_{2t-1})
= \frac{1}{\omega\sqrt{T}} \sum_{t=1}^T \xi_t.
\]
By construction, $(\xi_t)$ is an MDS with respect to $(\mathcal{F}_t)$ and satisfies $ E(\xi_t^2 \mid \mathcal{F}_{t-1}) = v_t^2$ under Assumption~\ref{assumption-mds}.  
The desired convergence then follows directly from the martingale central limit theorem, given Assumptions~\ref{assumption-1-2} and~\ref{assumption-2-1}.
\end{proof}

\section*{References}

\noindent Choi, Y., Jacewitz, S. and Park, J. Y. (2016), `A reexamination of stock return predictability’,
\emph{Journal of Econometrics} \textbf{192}, 168-189.\bigskip

\noindent Demetrescu, M., Georgiev, I., Rodrigues, P. M. and Taylor, A. R. (2023), `Extensions to
ivx methods of inference for return predictability’, \emph{Journal of Econometrics} \textbf{237} (2, Part
C), 105271.\bigskip

\noindent Hall, P. and Heyde, C. C. (1980), Martingale Limit Theory and Its Application, Academic
Press, New York.\bigskip

\noindent Hansen, B. E. (1992), `Convergence to stochastic integrals for dependent heterogeneous processes’, \emph{Econometric Theory} \textbf{8}, 489-500.\bigskip

\noindent Ibragimov, R., Kim, J. and Skrobotov, A. (2024), `New robust inference for predictive regressions’, \emph{Econometric Theory} \textbf{40}, 1364-1390.\bigskip

\noindent Ibragimov, R. and M¨uller, U. (2010), `\textit{t}-statistic Based Correlation and Heterogeneity Robust
Inference’, \emph{Journal of Business \& Economic Statistics} \textbf{28}, 453-468.\bigskip

\noindent Ibragimov, R. and Phillips, P. C. B. (2008), `Regression asymptotics using martingale convergence methods’, \emph{Econometric Theory} \textbf{24}, 888-947.\bigskip

\noindent Kostakis, A., Magdalinos, T. and Stamatogiannis, M. (2015), `Robust econometric inference
for stock return predictability’, \emph{Review of Financial Studies} \textbf{28}, 1506-1553.\bigskip

\noindent Kostakis, A., Magdalinos, T. and Stamatogiannis, M. P. (2023), `Taking stock of longhorizon predictability tests: Are factor returns predictable?’, \emph{Journal of Econometrics}
\textbf{237}(2), 105380.\bigskip

\noindent Kurtz, T. G. and Protter, P. (1991), `Weak limit theorems for stochastic integrals and stochastic
differential equations’, \emph{Annals of Probability} \textbf{19}(3), 1035-1070.\bigskip

\noindent Phillips, P. C. B. and Magdalinos, T. (2007), Limit theory for moderate deviations from a unit
root under weak dependence, in G. D. A. Phillips and E. Tzavalis, eds, `The Refinement
of Econometric Estimation and Test Procedures: Finite Sample and Asymptotic Analysis’,
Cambridge University Press, pp. 123-162.\bigskip


\end{document}